\documentclass[twocolumn,pra,showkeys,showpacs,groupedaddress]{revtex4-1}

\RequirePackage[colorlinks=true,linkcolor=blue,urlcolor=blue,citecolor=blue]{hyperref}
\usepackage{comment}
\usepackage{amsthm}
\def\dfn{\mathrel{\mathop:}=}

\def\bh{\mathbf{h}}
\def\<{\langle}\def\>{\rangle}

\def\Cmplx{\mathbb{C}}\def\Reals{\mathbb{R}}\def\N{\mathbb{N}}
\usepackage{hyperref}
\usepackage{enumerate}
\usepackage[inline]{enumitem}

\def\bvec#1{\mathbf{#1}}
\def\bk{\bvec k}
\def\bh{\bvec h}
\def\bx{\bvec x}

\usepackage{graphicx}
\usepackage[matrix,frame,arrow]{xypic}
\usepackage[font=small,labelfont=bf,
justification=raggedright,
format=plain]{caption}
\usepackage[font=small,labelfont=bf,
justification=justified,
format=plain]{subcaption}
\newcommand{\kket}[2]{\left| #1 \vphantom{#2} \right>\! \left| #2 \vphantom{#1} \right>}

\newcommand{\kbra}[2]{\left| #1 \vphantom{#2} \right>\! \left< #2 \vphantom{#1} \right|}

\theoremstyle{plain}
\newtheorem*{definition*}{Definition}
\newtheorem{definition}{Definition}
\newtheorem{theorem}{Theorem}
\newtheorem*{theorem*}{Theorem}
\newtheorem{corollary}{Corollary}
\newtheorem*{corollary*}{Corollary}
\newtheorem*{proposition*}{Proposition}
\newtheorem{proposition}{Proposition}

\newtheorem*{conjecture*}{Conjecture}
\newtheorem*{question*}{Question}

\newtheorem*{problem*}{Problem}

\newtheorem*{lemma*}{Lemma}
\newtheorem{lemma}{Lemma}
\newtheorem*{example*}{Example}
\newtheorem{example}{Example}
\theoremstyle{remark}

\newtheorem*{property*}{Property}
\newtheorem*{remark*}{Remark}
\newtheorem{remark}{Remark}

\usepackage{graphicx}
\usepackage{amsthm}
\usepackage{amssymb,amsmath,mathrsfs,color}
\usepackage{epsfig}
\usepackage[matrix,frame,arrow]{xypic}
\usepackage{mathtools}
\usepackage{dsfont}
\usepackage{comment}
\usepackage[font=small,labelfont=bf,
   justification=justified,
   format=plain]{caption}
\usepackage[font=small,labelfont=bf,
justification=justified,
format=plain]{subcaption}
\renewcommand{\v}[1]{\ensuremath{\mathbf{#1}}}
 
\newcommand{\ket}[1]{\left| #1 \right>}
\newcommand{\bra}[1]{\left< #1 \right|}
\newcommand{\braket}[2]{\left< #1 \vphantom{#2} \right|
 \left. #2 \vphantom{#1} \right>}
\def\bvec#1{\mathbf{#1}}
\def\bk{\bvec k}
\def\bh{\bvec h}

\def\Proof{\medskip\par\noindent{\bf Proof.  }}
\def\qed{$\,\blacksquare$\par} \def\>{\rangle} \def\<{\langle}

  \def\Reals{{\mathbb R}}

\usepackage{soul}

\begin{document}
\title{Chirality from quantum walks without quantum coin}
\date{\today}
\author{Giacomo Mauro D'Ariano}
\email{dariano@unipv.it}
\author{Marco Erba}
\email{marco.erba@unipv.it}
\author{Paolo Perinotti}
\email{paolo.perinotti@unipv.it}
\affiliation{Universit\`a degli Studi di Pavia, Dipartimento di Fisica, QUIT Group, and INFN Gruppo IV, Sezione di Pavia, via Bassi 6, 27100 Pavia, Italy}
\begin{abstract}
Quantum walks (QWs) describe the evolution of quantum systems on graphs. An intrinsic degree of freedom---called the coin and represented by a finite-dimensional Hilbert space---is associated to each node.
Scalar quantum walks are QWs with a one-dimensional coin. 
We propose a general strategy allowing one to construct scalar QWs on a broad variety of graphs, which admit embedding in Eulidean spaces, thus having a direct geometric interpretation. After reviewing the technique that allows one to regroup cells of nodes into new nodes, transforming finite spatial blocks into internal degrees of freedom, we prove that no 
QW with a two-dimensional coin can be derived from an isotropic scalar QW in this way. Finally we show that the Weyl and Dirac QWs can be derived from scalar QWs in spaces of dimension up to three, via our construction.

\end{abstract}
\pacs{02.10.Ox, 02.10.Yn, 02.20.-a, 03.65.Ta, 03.65.Pm, 03.67.-a, 03.67.Ac}
\keywords{Quantum walks, Cayley graphs, coinless QWs, Group extensions, Weyl QW, Dirac QW, Chirality}
\maketitle

\section{Introduction}
Quantum walks (QWs) on graphs~\cite{aharonov2001quantum,ambainis2001one,S03,M07,Venegas-Andraca2012} describe the evolution of quantum systems in a discrete arena. Their applications range from quantum information and computation~\cite{childs2003exponential,
ambainis2007quantum,magniez2007quantum,PhysRevA.81.042330}---where they have been applied to design search algorithms~\cite{ambainis2003quantum,shenvi2003quantum,kendon2006random,PhysRevLett.102.180501,PhysRevA.79.012325,portugal2013quantum,Childs791}---to discrete approaches for the foundations of relativistic Quantum Field Theory (QFT)~\cite{PhysRevA.90.062106,Bisio2017,ArrighiNJP14,PhysRevA.93.052301,PhysRevA.94.042120,PhysRevA.97.032132}---where QWs are particularly suitable for a reformulation of free QFTs in a discrete scenario~\cite{arrighi2013decoupled,arrighi2014dirac,B2015,PhysRevA.94.012335,bisio2016quantum,DiMolfetta2016,PhysRevA.97.042131}.
In particular, foundational investigations have been carried out focusing on how a continuous space-time, fermionic dynamics and Lorentz symmetry can be reconstructed by a discrete and purely quantum theory~\cite{BBDPT15,Bisio2015,Bisio20150232,d2017isotropic,brun2018detection}. In Ref.~\cite{dariano2016} it has been proved that, under the physical assumption of homogeneity of the physical evolution \footnote{This assumptions require not only that the QW evolution rule is the same at each node, but also that the nodes are not distinguishable on the basis of the evolution, i.e.~that the if a labeled path is closed starting from a given node, then it is closed starting from any other node.}, the graph of the QW must be in fact a \emph{Cayley graph}, namely the graphical representation of a group. This allows one to exploit the group-theoretical machinery, which is of aid in order to construct the QWs, analyze them and connect the graph to the emergent continuous geometry~\cite{acevedo2006quantum,PhysRevA.90.062106,BDEPT16,DEPT16}.

Historically~\cite{MD96}, QWs have been introduced in the broader context of Quantum Cellular Automata (QCAs)~\cite{schumacher2004reversible}, which provide a general model for the local unitary evolution of quantum systems on arbitrary graphs. In the case where the evolution law is linear in the fields, a QCA reduces to a QW, representing the quantum counterpart of the classical random walk model. 

So far, the case of a Euclidean emergent space has been particularly studied, as a simplifying restriction of the theory. It is worth mentioning here that, in fact, non-Euclidean cases have been hardly treated in the context of QWs, while the literature focuses on QWs on lattices. The main reason to focus on the Euclidean case is that lattices have a convenient embedding in the usual space $\mathbb{R}^d$. Furthermore, this has the advantage of allowing one to easily define the Fourier transform on the lattice and analyze the QWs dynamics in the wave-vector space, and this in turn provides a straightforward procedure for taking the continuum limit. Restricting to Cayley graphs which have a Euclidean emergent space $\mathbb{R}^d$ is equivalent to consider groups containing finitely many copies of $\mathbb{Z}^d$~\cite{CTV07}: these groups are called \emph{virtually Abelian}~\cite{DEPT16}.

A QW is, loosely speaking, a unitary evolution for a wave-function on a graph. Each vertex of the graph is structured as a local quantum system, and thus associated with a finite-dimensional Hilbert space, usually referred to as the {\em coin system of the walk at a given vertex}. If the QW is homogeneous, every vertex is equivalent, and thus all the coin systems are isomorphic to a prototype Hilbert space called the {\em coin system} of the QW. One can now constructively analyze all the conceivable QWs on a given graph. This analysis can be carried out in a twofold way: on the one hand, one can fix the simplest allowed group, i.e.~$\mathbb{Z}^d$ itself, and investigate the admissible QWs defined on its Cayley graphs, varying the coin dimension $s$; on the other hand, one can fix $s$ and construct all the admissible groups and graphs complying to the Euclidean restriction. The first path has been carried out in Refs.~\cite{BDEPT16,d2017isotropic}.

In this manuscript, we explore the second way, starting with the minimal coin dimension $s=1$, and performing a systematic analysis of the Euclidean scenario. QWs with a one-dimensional coin are often referred to as \emph{scalar} or \emph{coinless}~\cite{PhysRevA.71.032347,ARC08,PhysRevA.91.052319,Santos:2015:MCQ:2822142.2822148,BDEPT16}. Despite the algorithmic simplicity of the model, finding all scalar QWs for an arbitrary graph is not a straightforward task. Indeed, the resolution of the unitarity constraints involves a quadratic system of complex equations, and it turns out that it is simpler to address the problem as a matrical one. In Ref.~\cite{BDEPT16}, scalar QWs on Cayley graphs of arbitrary Abelian groups have been classified, finding that they give rise to trivial dynamics: this result extends the classical no-go theorem by Meyer~\cite{MD96}. Therefore, the present approach is to undertake a systematic investigation of QWs on Cayley graphs of virtually Abelian groups, relaxing the Abelianity assumption. This allows to have nontrivial dynamics.

The present manuscript benefits from the work carried out in Refs.~\cite{DEPT16,BDEPT16}, and represents their completion. Here we construct the first examples of infinite scalar QWs with more than one space-dimension; moreover, we avoid partitions of the underlying graph or of the QW itself, as opposed to the literature where the partition is used~\cite{PhysRevA.71.032347,PhysRevA.91.052319,Santos:2015:MCQ:2822142.2822148} to circumvent the no-go theorem~\cite{BDEPT16}, leading to an inhomogeneous evolution. Furthermore, the scalar QWs constructed here are both non-Abelian and infinite, unlike those explored in Ref.~\cite{ARC08}, where only finite graphs have been considered.

The manuscript is organized as follows. In Sec.~\ref{sec:DTQW} we review the general model of discrete-time QWs on graphs; we then specialize our treatment to the Euclidean case, establishing a connection between algebraic and geometrical properties of groups. In Sec.~\ref{sec:extension_problem} we study the group extension problem in generality, proving some structure results and reviewing a coarse-graining technique for QWs on Cayley graphs; we then apply the aforementioned results and technique to the Euclidean case. In Sec.~\ref{sec:euclidean_scalar_QWs} we investigate the Euclidean scalar QWs, proving a no-go theorem in the isotropic case, and then paving the way for the study of a particular class of QWs in space-dimension $d=1,2,3$. In Sec.~\ref{sec:relativistic_equations} we apply the group extension technique to particularly analyze the case of scalar QWs whose coarse-grainings are QWs with a two-dimensional coin on the simple square and on the BCC lattices; finally we derive the Dirac QW in two and three space-dimensions as the coarse-graining of a scalar QW on the Cayley graph of a non-Abelian group. In Sec.~\ref{sec:conclusions} we conclude the manuscript discussing some (open) aspects of the theory and drawing our conclusions.

\section{Discrete-time quantum walks on graphs}\label{sec:DTQW}

For the convenience of the reader, we start recalling the basic notion of a directed graph.

\begin{definition}
	A \emph{directed graph} or \emph{digraph} is an ordered pair $\Gamma = (V,E)$, where $V$ and $E$ are set such that $E$ collects arbitrary ordered pairs $(x_1,x_2)$ of elements $x_1,x_2\in V$. The elements in $V$ are called the \emph{vertices} and those in $E$ the \emph{edges} of the graph.
\end{definition}
The vertices are graphically represented as dots and are also called the \emph{sites} or \emph{nodes}. The set $E$ of edges defines the connectivity between the vertices of the graph: an edge $(x_1,x_2)$ is graphically represented by an arrow having direction from $x_1$ to $x_2$. In the following we define the neighbourhood schemes for the sites of a graph.
\begin{definition}
	Let $\Gamma = (V,E)$ be a graph. We define the \emph{first-neighbourhood} of each site $x\in V$ as the set
	\begin{equation*}
		N_x \coloneqq \lbrace y\in V \left. | \right. (x,y)\in E \rbrace,
	\end{equation*}
	namely all the vertices reached by arrows from $x$. The elements contained in $N_x$ are called first-neighbors of $x$. The \emph{complement of the first-neighborhood} of each site $x\in V$ is the set defined as
	\begin{equation*}
	N_x^{-1} \coloneqq \lbrace y\in V \left. | \right. (y,x)\in E \rbrace .
	\end{equation*}
\end{definition}
We are now ready to give the definition of a discrete-time QW on a graph.
\begin{definition}
	\label{def:qw}
	Let $\Gamma = (V,E)$ be a graph and let a finite-dimensional Hilbert space $\mathscr{H}_x$ be associated to each node $x\in V$. A QW on $\Gamma$ in $\mathscr{H}\coloneqq \bigoplus_{x\in V} \mathscr{H}_x$ is a unitary operator $W$ providing a time-homogeneous evolution defined as follows:
	\begin{align*}
	W\ :\ \mathscr{H} \ &\longrightarrow \ \mathscr{H} \\
	\ket{\psi(t)}\ &\longmapsto\ \ket{\psi(t+1)}
	\end{align*}
	for all times $t$, such that, defining $\Pi_x$ as the projection on $\mathscr H_x$, one has
	\begin{align*}
	\Pi_xW:
		\bigoplus_{y\in N_x^{-1}}\ket{\psi_y(t)}\     &\longmapsto \    \ket{\psi_x(t+1)}\quad \forall x\in V.
	\end{align*}
\end{definition}
This defines a discrete-time evolution on a graph $\Gamma$ according to its neighbourhood schemes. By linearity of the operator $W$, 
one can block-decompose the evolution as:
	\begin{align}\label{eq:qwalk}
	\begin{split}	
&\Pi_xW\bigoplus_{y\in N_x^{-1}}\ket{\psi_y(t)} = \sum\limits_{y\in N_x^{-1}}A_{yx}\ket{\psi_y(t)},
	\end{split}
	\end{align}
	where the $A_{xy}$ are $\text{dim}\mathscr{H}_x \times\text{dim}\mathscr{H}_y$ matrices, called the \emph{transition matrices} of the QW.
Definition~\ref{def:qw} represents the general definition of a QW, as originally given in Ref.~\cite{MD96}, namely every QW model admits a form~\eqref{eq:qwalk}. We now wish to represent the QW evolution on the total Hilbert space
\begin{align*}
\mathscr{H}_{\mathrm{tot}}\coloneqq\bigoplus_{x\in V}\ket{x}\otimes \mathbb{C}^{s_x}\cong \mathscr{H}.
\end{align*}
One can verify that the action of the QW evolution is represented on $\mathscr{H}_{\mathrm{tot}}$ by the following operator:
\begin{align}\label{eq:general_qwoperator}
A = \sum_{x\in V}\sum_{y\in N_{x}} \Delta_{xy}\otimes A_{xy} \coloneqq \sum_{x\in V}\sum_{y\in N_{x}} \kbra{y}{x}\otimes A_{xy}.
\end{align}
We notice that the evolution can be rewritten in terms of the first-neighborhoods, while in expression~\eqref{eq:qwalk} the sum is on their complements. Thus the walk operator can be finally written in terms of the  edges of the graph as:
\begin{equation}\label{eq:woperator}
	A = 
	\sum_{x\in V} \sum_{f\in D_{x}} \Delta_{xx_f} \otimes A_{xx_f},
\end{equation}
where the set $D_x$ collects the edges $f$ connecting $x$ to its first-neighbors $x_f$ and the $\Delta_{xx_f}$ are the shift operators, mapping $\ket{x}$ to $\ket{x_f}$. Eq.~\eqref{eq:woperator} represents the most general form of a QW evolution operator. The graph structure is inferred from the neighborhood schemes $N_x$ and the evolution may also be inhomogeneous in the sites. In the following, as already discussed, we will restrict to homogeneous QWs, namely such that one has: (i) $D\coloneqq D_{x}=D_{x'}$, (ii) $|N_x|=|N_{x'}|,|N_x^{-1}|=|N_{x'}^{-1}|$, and (iii) $A_{xx_f}=A_{x'x'_f}$ for all $x,x'\in V$ and $f\in D$. Accordingly, we consider \emph{regular directed graphs}. Moreover, the edges are equipped with the same set of associated transition matrices. Thus, the dimension of the coin system is taken to be the same at any vertex---say $s$.
Accordingly, $\mathscr{H}_{\mathrm{tot}}= \ell^2(V)\otimes \mathbb{C}^{s}$, and one has the following form for the evolution operator:
\begin{align}\label{eq:woperator2}
\begin{split}
A = \sum_{f\in D} \Delta_f\otimes A_f \coloneqq \sum_{f\in D} \left(\sum_{x\in V} \kbra{x_f}{x}\right) \otimes A_f,
\end{split}
\end{align}
where $D$ collects the set of edges, $x_f$ represents the first-neighbor of $x$ connected by the edge labeled as $f$, and the $A_f$ are $s\times s$ complex matrices.
\begin{remark}\label{rem:1}
Equations~\eqref{eq:qwalk} and~\eqref{eq:general_qwoperator} are equivalent expressions providing the evolution of a QW in terms of the transition matrices. In the literature, it is quite common to present QWs in terms of a decomposition into a product of (generally unitary) operators, namely in the following form:
\begin{align}\label{eq:wdecomposition}
\begin{split}
\ket{\Psi(t+1)} &= \sum_{x\in V}\kket{x}{\psi_x(t+1)} = A\ket{\Psi(t)} = \\ 
&=\sum_{x\in V}O_1O_2\cdots O_n \kket{x}{\psi_x(t)}.
\end{split}
\end{align}
In general, these operators take the form
\begin{align*}
O_l = \sum_{i,j=1}^s c_{ij}^l S_{ij}^l \otimes \kbra{i}{j},
\end{align*}
where $c_{ij}^l\in \mathbb{C}$, $\lbrace \ket{i} \rbrace_{i=1}^s$ is the canonical basis of $\mathbb{C}^s$, and the $S_{ij}^l$ act as shift operators, namely
\begin{align*}
S_{ij}^l\ket{x} = \ket{\pi_{ij}^l(x)}
\end{align*}
for some permutation $\pi_{ij}^l$ (including the identical one) of the vertices of an underlying graph. Examples of decomposition~\eqref{eq:wdecomposition} can be found in Refs.~\cite{arrighi2018dirac,costa2018,jay2018dirac}. It is common to associate a geometrical meaning to the shift operators appearing in the definition of the operators $O_l$, namely to infer the graph structure from these. Yet this interpretation is in some sense misleading, since, for a fixed QW, the decomposition in the form~\eqref{eq:wdecomposition} may not be unique. Consequently, the same QW might be decomposed into ``virtual'' steps, as in~\eqref{eq:wdecomposition}, in several ways. On the other hand, all such different decompositions correspond to a unique expression of the form~\eqref{eq:general_qwoperator}. Moreover, the ``virtual'' steps in~\eqref{eq:wdecomposition} are not elementary steps of the walk evolution: the elementary steps are in any case given by each application of the total walk operator $A=O_1O_2\cdots O_n$. In the context of homogeneous QWs, a paradigmatic example of this fact is given by the three-dimensional Dirac QW. In Ref.~\cite{arrighi2014dirac} one can find the Dirac QW expressed as a product of three unitaries, each involving a translation on one Cartesian axis, corresponding to a decomposition on a simple cubic lattice. The same QW can be expressed in the form~\eqref{eq:woperator2} on the BCC lattice. In two recent works~\cite{arrighi2018dirac,jay2018dirac} one can find the two-dimensional Dirac QWs decomposed into local operators in several ways, each one with a different ``virtual'' graph. Nevertheless, each of them is indeed a different decomposition of the same QW, namely the two-dimensional Dirac QW on the simple square lattice~\cite{PhysRevA.90.062106}. On the other hand, in Ref.~\cite{d2017isotropic} a theorem has been proven stating that, for isotropic QWs on lattices (graphs embeddable in $\mathbb R^d$) with a two-dimensional coin, expression~\eqref{eq:woperator2} identifies indeed a unique graph for $d\leq 3$ (namely the integer lattice in $d=1$, the simple square in $d=2$, and the BCC lattice in $d=3$). We point out that, however, the different equivalent decompositions~\eqref{eq:wdecomposition} have in fact a relevance with regard to the concrete implementation and simulation of the same QW.
\end{remark}

\subsection{Quantum walks on Cayley graphs}\label{sec:QW_on_Cayley}
In the present subsection we shall treat QWs on Cayley graphs in full generality. In Subsec.~\ref{subsec:embedding} and from Subsec.~\ref{subsec:extensionsZd} onwards, we shall restrict our attention to Cayley graphs of virtually Abelian groups (i.e.~the Euclidean case). 
Throughout this paper, we will write $G_1 \leq G_2$ if $G_1$ is a subgroup of $G_2$, and $G_1\trianglelefteq G_2$ if $G_1$ is normal in $G_2$. Finally, $Z(G)$ will denote the \emph{center} of $G$.

Consider a finitely generated group
\begin{align}\label{eq:presentation}
G = \<S_+|R\>,
\end{align}
where $S_+$ is a (finite) \emph{generating set} for $G$, and $R$ a set of \emph{relators}. We will denote the set of inverses of the elements in $S_+$ by $S_-$. A generating set $S_+$ is called \emph{symmetric} if $S_+=S_-$ and we define the \emph{set of generators} $S\coloneqq S_+\cup S_-$ of $G$, which is clearly symmetric. Every group element $g\in G$ is a word defined on the alphabet $S$. On the other hand, $R$ is a set of closed paths generating all the cycles in the group by concatenation or conjugation with arbitrary words. The closed paths correspond to words of $S$ which amount to the identity element $e\in G$. 

Expression~\eqref{eq:presentation} is called a \emph{presentation} of $G$. In the following, we will restrict our attention to finitely presented groups, namely such that also $|R|<\infty$ holds. Every group can be presented, in principle, in infinitely many ways, the presentations being in one-to-one correspondence with \emph{Cayley graphs}, as it is clear from the following definition.
\begin{definition}
The Cayley graph $\Gamma(G,S_+)$ of a group $G$ with respect to the generating set $S_+$ is the edge-coloured directed graph constructed as follows: (i) $G$ is the vertex set of $\Gamma$; (ii) for all $g\in G$, a coloured edge directed from $g$ to $gh$ is assigned to each $h\in S_+$.
\end{definition}
Edges corresponding to some $h\in S_+$ are represented as undirected if and only if $h^2 = e$. In general we allow $e\in S_+$, whose corresponding edges can be denoted by loops on the Cayley graph.

We shall consider the right-regular representation $T$ of a group $G$ on $\ell^2(G)$, given as follows. We will denote
by $\{ \ket{g} \}_{g\in G}$ the canonical basis for $\ell^2(G)$ and we define
\begin{equation}
T_{g'}\ket{g}\coloneqq \ket{gg'^{-1}},
\end{equation}
from which it follows that $T_gT_{g'}=T_{gg'}$ for all $g,g'\in G$. By construction, the right-regular representation is unitary. For finite $G$ one has $\ell^2(G) \equiv \Cmplx^{|G|}$.
\begin{definition}
	Let $G$ be a finitely presented group. A quantum walk on the Cayley graph $\Gamma(G,S_+)$ with an $s$-dimensional coin is the quadruple
\begin{equation*}
W=\{G,S,s,\{A_h\}_{h\in S}\},
\end{equation*}
such that:
\begin{enumerate}
\item $s\in\N^+$;
\item for all $h\in S$, the \emph{transition matrices} $A_h\in\mathrm{M}_s(\Cmplx)$;
\item the operator
\begin{equation}\label{walkop}
A = \sum_{h\in S}T_h\otimes A_h
\end{equation}
defined on $\ell^2(G)\otimes\Cmplx^s$ is unitary.
\end{enumerate}
\end{definition}
The walk operator~\eqref{walkop} is unitary if and only if all the following system of equations holds:
\begin{align}\label{unitarity}
\begin{split}
\sum_{\substack{h,h'\in S\colon \\h{h'}^{-1}=g}}A_{h}A_{h'}^\dag = \sum_{\substack{h,h'\in S\colon \\h^{-1}h'=g}}A_h^\dag A_{h'} =\delta_{g,e}I_s, \\
\forall g\in \lbrace g'\in G \left.|\right. \exists h_1,h_2\in S \colon g'=h_1h_2^{-1} \rbrace .
\end{split}
\end{align}
Eq.~\eqref{unitarity} can be checked just plugging expression~\eqref{walkop} into the unitarity conditions:
\begin{equation*}
A^{\dagger}A = AA^{\dagger} = T_e \otimes I_s.
\end{equation*}
Given a Cayley graph $\Gamma$, the unitarity conditions~\eqref{unitarity} represent nontrivial constraints to solve in order to define a QW on $\Gamma$.

The following definition is useful to introduce the generalized notion of \emph{isotropy} for QWs~\cite{PhysRevA.90.062106,d2017isotropic}. This feature is relevant to model the physical law in a theory aiming to reconstruct relativistic QFTs. Given a group $G$, the \emph{order of an element} $g\in G$ is defined as the natural number
\begin{equation*}
r_{g} \dfn \min \lbrace r\in \N^+ \colon g^r = e \rbrace .
\end{equation*}
If $r_g$ does not exist, the order of $g$ is said to be infinite. We are now ready to define an isotropic QW.
\begin{definition}
	\label{isotropy}
A QW on the Cayley graph $\Gamma(G,S_+)$ with a $s$-dimensional coin is called isotropic if there exists a faithful unitary representation $\{U_l\}_{l\in L}$ on $\Cmplx^s$ of a graph automorphism group $L$, called the \emph{isotropy group}, satisfying the following:
\begin{enumerate}
\item $L$ is transitive on the classes of elements of $S_+$ having the same order~\footnote{In the original definition~\cite{PhysRevA.90.062106}, where just free Abelian groups were considered, this requirement has been stated without mentioning the possibility of different element-orders.};
\item the following invariance condition holds:
\begin{equation}\label{eq:isotropic_invariance}
\sum_{h\in S}T_{l(h)}\otimes U_l A_{h} U_l^\dag=\sum_{h\in S}T_{h}\otimes A_{h},\quad \forall l\in L;
\end{equation}
\end{enumerate}
\end{definition}
In the case of a free QFT derived as a QW theory, the following additional requirement has been demanded:
\begin{align}\label{eq:isotropic_commutation}
[U_l , A_{h}] \neq 0\quad \forall h\in S,l\in L\colon l(h)\neq h.
\end{align}
Condition~\eqref{eq:isotropic_commutation} is never satisfied in the scalar case, since both the representation $\{U_l\}_{l\in L}$ and the transition matrices are one-dimensional. However, as long as the Cayley graph structure is not derived by the requirement of homogeneity of the evolution---but rather assumed as in the present context---condition~\eqref{eq:isotropic_commutation} is dropped. In Subsec.~\ref{subsec:isotropic_scalar_QWs} we will prove a no-go result for isotropic scalar QWs without assuming condition~\eqref{eq:isotropic_commutation}.
This shows that in the Euclidean case, no QW with a two-dimensional coin can be derived from an isotropic scalar QW, even using the weak Definition~\ref{isotropy} of isotropy.

In Ref.~\cite{d2017isotropic} it has been proven that, in the case of QWs on lattices (namely Cayley graphs of $\mathbb{Z}^d$), isotropy entails that all the generators can be represented with the same length in $\mathbb{R}^d$. In particular, this implies that one has the following unitarity constraints:
\begin{align}\label{eq:unitarity_2h}
A_{h}A_{-h}^\dagger = A_{h}^\dagger A_{-h} = 0,\quad \forall h\in S.
\end{align}
Eq.~\eqref{eq:unitarity_2h} implies that the transition matrices assume the following form~\cite{DEPT16}:
\begin{align}\label{eq:transmatr_form}
A_{\pm h} = \alpha_{\pm h} V_h \ket{\eta_{\pm h}}\bra{\eta_{\pm h}},
\end{align}
where, for all $h\in S$, $\alpha_{\pm h}>0$, $V_h$ is unitary, and $\lbrace \ket{\eta_{h}}, \ket{\eta_{-h}} \rbrace$ is an orthonormal basis in the coin space for every $h$. Finally, in Ref.~\cite{PhysRevA.90.062106} it is shown that in the case of Abelian $G$ one 
has $\sum_{h\in S} A_h =U$ for some unitary $U$ commuting with the isotropy group representation $\lbrace U_l \rbrace_{l\in L}$. The same happens for arbitrary group $G$, by a straightforward generalisation.  Therefore, for the purpose of classification one can first solve the unitarity and isotropy constraints by posing
\begin{align}\label{eq:sum_unit}
\sum_{h\in S} A_h = I_s,
\end{align}
and then obtain the other solutions upon multiplying all the matrices $A_h$ by an arbitrary unitary $U$ in the commutant of $\lbrace U_l \rbrace_{l\in L}$. We shall make use of Eqs.~\eqref{eq:unitarity_2h},~\eqref{eq:transmatr_form}, and~\eqref{eq:sum_unit} in the proof of our main results in Secs.~\ref{sec:euclidean_scalar_QWs} and~\ref{sec:relativistic_equations}.

We conclude this section recalling a necessary condition for the existence of a scalar (or coinless) QW on a given Cayley graph.
\begin{proposition}[Quadrangularity condition~\cite{ARC08}]\label{quadrang}
	Given a Cayley graph $\Gamma(G,S_+)$, a necessary condition for the existence of a scalar QW
	\begin{equation*}
	A = \sum_{h\in S} z_hT_h
	\end{equation*}
	is that, for all the ordered pairs $(h_1,h_2) \in S\times S$ such that $h_1 \neq h_2$, there exists at least a different pair $(h_3,h_4)$ such that $h_1h_2^{-1}=h_3h_4^{-1}$. This is called \emph{quadrangularity condition}.
\end{proposition}
\begin{remark}
Proposition~\ref{quadrang} holds for homogeneous coinless QWs, according to the definition of the present work and related literature (see e.g.~\cite{ARC08,BDEPT16}); yet, a similar result must hold, for any two connected nodes, also for inhomogeneous QWs. Elsewhere, the homogeneity requirement has been in general dropped (see e.g.~Refs.~\cite{PhysRevA.95.012328,costa2018} for the model of the so called Staggered QW with Hamiltonians, or SQWH). It is interesting to notice that---imposing homogeneity to the transition scalars in the 1D example of SQWH presented in Ref.~\cite{costa2018}---by direct inspection of the first-neighborhood scheme one easily realizes that this QW is contained in the family of QWs on the infinite dihedral group found in~\cite{BDEPT16}. Again, from the point of view of a decomposition in the form~\eqref{eq:wdecomposition}, the ``virtual'' graph can be regarded as the integer lattice, while expressing the QW in the form~\eqref{eq:woperator2} (which is unique), one realizes that the actual graph is indeed a Cayley graph of the infinite dihedral group. Remarkably, it is possible that the present model is actually contained in the SQWH model. However, both constructions have been devised in order to overcome the issue of constructing coinless QWs. A no-go theorem proven in Ref.~\cite{BDEPT16}, and generalizing a result by Meyer~\cite{MD96}, states that coinless QWs on Abelian graphs exhibit trivial dynamics. In the present work we show how to overcome this issue considering more general graphs, while keeping a homogeneous evolution rule. We point out that generally inhomogeneous models (like the Staggered QW one) find their relevance on the side of the experimental implementation, or also to mimic a curved space-time~\cite{arrighi2018curved}.
\end{remark}

\subsection{The continuum limit}
Let $G$ be isomorphic to $\mathbb{Z}^d$, for an arbitrary integer $d\geq1$. Let then $\Gamma(G,S_+)$ be a Cayley graph of $G$. This case encompasses the usual lattices, including most of the cases treated in the literature where QWs are exploited to simulate wave equations or to devise algorithms. In particular, in this case the shift operators commutes. Accordingly, throughout the present subsection we will use Abelian notation, denoting the group composition law with the additive notation and the elements $\v{x}\in G$ as boldfaced $d$-dimensional real vectors. In particular, we will consider $G$ as a space-vector.

Now, let the unitary operator
\begin{equation*}
A =  \sum_{\bh \in S} \left( \sum_{\v{y}\in G} \kbra{\v{y}-\bh}{\v{y}} \right) \otimes A_\bh =  \sum_{\bh \in S} T_\bh \otimes A_\bh
\end{equation*}
represent a QW on the Cayley graph $\Gamma$. The unitary irreducible representations of $\mathbb{Z}^d$ are one-dimensional. We now constructively show how to decompose the right-regular representation, which is unitary by definition, into the one-dimensional unitary irreducible representations of $G$. The latter are classified by the joint eigenvectors given by the relation
\begin{equation}\label{eq:right_reg_irrep}
T_{\v{y}} \ket{\bk} = e^{i\bk\cdot \v{y}} \ket{\bk},\quad \v{y}\in G,
\end{equation}
where
$\bk$ is an element of the dual $G^*$. Taking the following expansion for the eigenvectors
\begin{equation}\label{eq:wave-vector}
\ket{\bk} = \sum_{\bx \in G} c(\bx,\bk) \ket{\bx},
\end{equation}
and substituting it into Eq.~\eqref{eq:right_reg_irrep}, one obtains
\begin{align*}
T_{\v{y}} \ket{\bk} &= \sum_{\bx \in G} c(\bx,\bk) \ket{\bx-\v{y}}= \\ &= \sum_{\bx \in G} c(\bx+\v{y},\bk) \ket{\bx} = \sum_{\bx \in G} e^{i\bk\cdot \v{y}}c(\bx,\bk) \ket{\bx}.
\end{align*}
Accordingly, the relation $e^{-i\bk\cdot\v{y}} c(\bx+\v{y},\bk)=c(\bx,\bk)$ leads to $e^{i\bk\cdot\bx}c(0,\bk)=c(\bx,\bk)$.
Substituting the latter into Eq.~\eqref{eq:wave-vector} and imposing the normalization for the $\ket{\bk}$, we obtain:
\begin{align}\label{eq:eigenk}
\begin{split}
\ket{\bk} &= \frac{1}{(2\pi)^{d/2}} \sum_{\bx\in G} e^{i\bk\cdot \bx} \ket{\bx},\\ \ket{\bx} &= \frac{1}{(2\pi)^{d/2}} \int_B \! \textrm{d}\bk\; e^{-i\bk\cdot \bx} \ket{\bk},
\end{split}
\end{align}
where $B$ is the \emph{first Brillouin zone}, which we determine in the following. In general, the generators in $S_+$ are not linearly independent, then we define all the sets
\begin{equation*}
D_n\coloneqq\{\bh_{n_1},\ldots,\bh_{n_d}\}\subseteq S_+, 
\end{equation*}
collecting linearly independent elements, where $n$ labels the specific subset. For every $n$, we can then define the dual set
\begin{equation*}
\tilde D_n\coloneqq \{\tilde \bh_{n_1},\ldots,\tilde \bh_{n_d}\},\quad \tilde \bh_{n_l}\cdot \bh_{n_m}=\delta_{lm}.
\end{equation*}
We now can expand each $\bx\in G$ and $\bk\in G^*$ as
\begin{equation*}
\quad \bx = \sum_{j = 1}^d x_{n_j}\bh_{n_j},\quad \bk = \sum_{j=1}^d k_{n_j} \tilde{\bh}_{n_j},
\end{equation*}
for some $n$, where $x_{n_j} \in \mathbb{N}$ for every $j$ and $\bk\in B$. 
Two eigenstates $\ket{\bk},\ket{\bk'}$ are equivalent if there exists $\theta \in [0,2\pi]$ such that
\begin{equation*}
\ket{\bk} = e^{i\theta}\ket{\bk'}.
\end{equation*}
Thus, from Eq.~\eqref{eq:eigenk}, one can derive
\begin{equation*}
e^{-i(\bk-\bk')\cdot \bx} = e^{i\theta} = e^{-i(\bk-\bk')\cdot \v{y}},\quad \forall \bx,\v{y}\in G,
\end{equation*}
which is equivalent to the condition
\begin{equation*}
\exists \v{l}\in \mathbb{N}^d : k_{n_j} - k'_{n_j} = 2\pi l_j,\quad j=1,\ldots,d.
\end{equation*}
Since the choice of $D_n,\tilde D_n$ is arbitrary, defining $\tilde D:=\bigcup_n\tilde D_n$, the Brillouin zone $B\subseteq\mathbb{R}^d$ is the polytope defined as:
\begin{equation*}
B=\bigcap_{\tilde\bh\in\tilde D}\lbrace \bk\in\mathbb{R}^d\mid-\pi|\tilde\bh|^2\leq\bk\cdot\tilde\bh\leq\pi|\tilde\bh|^2\rbrace.
\end{equation*}
The evolution operator $A$ can be thus diagonalized as follows:
\begin{equation*}
A=\int_{B}\! \mathrm{d}\bk \; \kbra{\bk}{\bk} \otimes A_\bk,
\end{equation*}
where the the matrix
\begin{equation*}\label{Ak}
A_\bk\coloneqq\sum_{\bh\in S}e^{i\bh\cdot\bk}A_\bh
\end{equation*}
must be unitary for every $\bk$. Being $A_\bk$
polynomial in $e^{i\bh\cdot\bk}$, imposing unitarity straightforwardly amounts to find the same general set of constraints of Eqs.~\eqref{unitarity}.

Clearly, in general $A_\bk\in \mathbb{U}(s)$ and its eigenvalues are of the form	$e^{i\omega_l(\bk)}$, for some integer $1\leq l \leq s$. The functions in the set
\begin{equation*}
\{ \omega_1(\bk),\ldots , \omega_s(\bk) \}, \quad \bk\in B
\end{equation*}
are called the \emph{dispersion relations} of the QW. As one can realize from the unitarity constraints in Eqs.~\eqref{unitarity}, the operator $A$ is defined up to a global phase factor, and then in particular one can always choose $A_\bk\in \mathbb{SU}(s)$ without loss of generality. This fact, in the case $s=2$, implies that the dispersion relation of a QW with a two-dimensional coin system is of the form $\pm \omega(\bk)$. In particular, this are interpretable as the particle and antiparticle branches of the dispersion relation.

The Fourier representation allows one to define differential equations for the evolution of the eigenstates, and also study the continuum limit. Let us introduce the \emph{interpolating Hamiltonian} $H_I(\bk)$ defined by the relation:
\begin{equation*}
\mathrm{exp}(-iH_I(\bk)) \coloneqq A_\bk .
\end{equation*}
$H_I(\bk)$ generates a discrete-time unitary evolution interpolating through a continuous time $t$ as
\begin{equation*}
\mathrm{exp}(-iH_I(\bk)t)\ket{\psi(\bk,0)} = \ket{\psi(\bk,t)}.
\end{equation*}
Then we can write a Schr\"odinger-like differential equation
\begin{align}\label{eq:schrodinger_like}
i\partial_t \ket{\psi(\bk,t)} = H_I(\bk)\ket{\psi(\bk,t)}
\end{align}
and expand to the first order in $\bk$, obtaining
\begin{align}\label{eq:wave-equation}
\begin{split}
i\partial_t \ket{\psi(\bk,t)} =& \left[ H_I(\v{0})+ \left.\nabla_{\bk'} H_I(\bk')\right|_{\bk'=\v{0}}\cdot\bk\right] \ket{\psi(\bk,t)} + \\ &+ O(|\bk|^2) \ket{\psi(\bk,t)} .
\end{split}
\end{align}
Now, identifying $\bk$ with the momentum of the system, one can interpret Eq.~\eqref{eq:wave-equation} as a wave equation in the wave vector representation. Let us consider narrowband states $\ket{\psi(\bk,t)}$, where small wave vectors $|\bk| \ll 1$ correspond to small momenta for the system. Then, identifying the lattice step with an elementary invariant length (e.g.~a hypothetical Planck scale), the limit of small momenta is equivalent to the relativistic limit for the QW's evolution.

In the case $s=2$, let $\ket{u^\pm(\bk)}$ be the positive and negative frequency eigenstates of $H_I(\bk)$, namely such that: $H_I(\bk)\ket{u^\pm(\bk)} = \pm \omega(\bk) \ket{u^\pm(\bk)}$.
Then a so-called \emph{(anti)particle state} is defined as
\begin{equation*}
\ket{\psi^\pm(t)} = \int_B \!\frac{\mathrm{d}\bk}{(2\pi)^d} g(\bk,t) \kbra{u^\pm(\bk)}{u^\pm(\bk)}.
\end{equation*}
Taking the normalized distribution $g(\bk,t)$ smoothly peaked around a given $\bk_0\in B$,  the evolution given by Eq.~\eqref{eq:schrodinger_like} leads to a dispersive  Schr\"odinger dfferential equation for the QW:
\begin{equation}\label{eq:schrodinger_diff}
i\partial_t \tilde g(\v{x},t)=\pm \left[ \v{v}\cdot \nabla + \frac{1}{2}\v{D}\cdot \nabla \nabla \right] \tilde g(\v{x},t),
\end{equation}
where $\tilde g(\bx,t)$ is the Fourier transform of $e^{-i \bk_0\cdot\bx + i\omega(\bk_0)t}g(\bk,t)$. Eq.~\eqref{eq:schrodinger_diff} is a Fokker--Planck equation, with drift vector and diffusion matrix given respectively by
\begin{equation}
\v{v} = \left.\nabla_\bk \omega(\bk)\right|_{\bk=\bk_0},\quad \v{D} = \left.\nabla_\bk\nabla_\bk \omega(\bk)\right|_{\bk=\bk_0}.
\end{equation}
This method is general for the case of Cayley graphs of $\mathbb Z^d$, and in particular it has been exploited---e.g.~in Ref.~\cite{PhysRevA.90.062106}---to study the Weyl and Dirac QWs dynamics.

In Subsec.~\ref{subsec:coarse-graining}, via a unitary coarse-graining technique, we will prove that the Euclidean QWs are strictly contained in the QWs on $\mathbb{Z}^d$. Accordingly, one can apply the Fourier method and take the continuous limit for all Euclidean QWs, even in some particular non-Abelian cases, as showed in the following. This will allow us to reconstruct the Weyl and Dirac QWs from scalar QWs on non-Abelian groups.

To best of our knowledge, a technique allowing one to extend a similar Fourier method to the non-Euclidean case is still unknown. This would be relevant since it would allow one to study the continuum limit of QWs on graphs with a nonvanishing curvature, e.g.~on Fuchsian groups (which admit an embedding in the Poincar\'e disk). On the other hand, so far in the literature curvature has been implemented on classical gauge fields encoded in the transition matrices (see e.g.~Refs.~\cite{Arrighi2016,arrighi2018curved}).

\subsection{Embedding Cayley graphs into smooth manifolds}\label{subsec:embedding}
In this section we review some results from Geometric Group Theory \cite{DK17} connecting algebraic properties of groups to geometric ones. The starting point is recognizing that endowing Cayley graphs with the notion of a distance allows us to study them as metric spaces.
\begin{definition}
Let $G = \<S_+|R\>$ be a finitely generated group. The word length is the norm defined, for all $g\in G$, as
\begin{equation}
l_S(g) \dfn \min \lbrace n\in \N \ |\ g=h_{1}\cdots h_{n}, h_{i} \in S \rbrace .
\end{equation}
The norm $l_S$ induces the word metric, defined as
\begin{equation}
d_{G}^{(S)} (g,g') \dfn l_S(g^{-1}g')\quad \forall g,g'\in G.
\end{equation}
\end{definition}
We are interested in Cayley graphs suitably embeddable in a Euclidean space $\mathbb{R}^d$, with a notion of embedding resorting to the following concept of \emph{quasi-isometry}~\cite{DLH00}.
\begin{definition}
Let $(G,d_G)$ and $(M,d_M)$ be two metric spaces. A quasi-isometry is a function $\mathscr{E}:G\rightarrow M$ satisfying, for some fixed $a\geq 1$ and $b,c \geq 0$, and $\forall g,g'\in G$, the two following conditions:
\begin{align*}
&\frac{1}{a}d_G(g,g')-b \leq d_M(\mathscr{E}(g),\mathscr{E}(g')) \leq ad_G(g,g')+b, \\
&\forall m\in M\ \exists g\in G \colon d_M(m,\mathscr{E}(g)) \leq c.
\end{align*}
\end{definition}
The previous definition intuitively states that the two metrics are equivalent modulo fixed bounds. Quasi-isometry is an equivalence relation~\cite{campbell1999groups} and two metric spaces $G,M$ are called \emph{quasi-isometric} if there exists a quasi-isometry between them.
\begin{definition}
Let $\mathrm{P}$ be a group property. A group $G$ is called virtually $\mathrm{P}$ if there exists $H \leq G$ satisfying $\mathrm{P}$ and such that the cardinality of the coset space $|G/H|$ (called the \emph{index of }$H$\emph{ in }$G$) is finite.
\end{definition}
In the following  the property of a group of  ``being isomorphic to $G$'' will be denoted by the same symbol $G$.  For example ``the group $K$ is virtually $\mathbb Z$'' means 
that there exists a subgroup $H$ of $K$ isomorphic to $\mathbb Z$ with finite index in $K$.
The next definition further refines the notion of virtually $\mathrm{P}$ group and shall be useful to the characterization of the groups whose Cayley graphs are quasi-isometric to $\mathbb{R}^d$.
\begin{definition}
Let $\mathrm{N}$ and $\mathrm{Q}$ be two group properties. A group $G$ is called \emph{$\mathrm{Q}$-by-$\mathrm{N}$} if there exists $N \trianglelefteq G$ satisfying $\mathrm{N}$ and such that the quotient group $G/N$ satisfies $\mathrm{Q}$. 
\end{definition}
We now provide some useful results to the purpose of establishing a quasi-isometric equivalence between $\mathbb{R}^d$ and virtually $\mathbb{Z}^d$ groups.
\begin{theorem}[Fundamental theorem of finitely generated Abelian groups]\label{fundab}
Every finitely generated Abelian group is isomorphic to a direct product of finitely many cyclic groups.
\end{theorem}
%
%
\begin{lemma}
Let $G$ be a group and $\mathrm{P}$ a group property inherited by subgroups of finite index. Then $G$ is finite-by-$\mathrm{P}$ if and only if it is virtually $\mathrm{P}$.
\end{lemma}
\Proof
($\Rightarrow$) It follows by definition.

($\Leftarrow$) Let $H$ be of finite index in $G$ and satisfying $\mathrm{P}$. Let us define
\begin{align*}
N_H \dfn \bigcap_{g\in G}gHg^{-1},
\end{align*}
namely the {\em normal core} of $H$ in $G$. Clearly, $N_H \trianglelefteq H$. Furthermore, by a result due to Poincar\`e~\cite{ST12}, $|G/N_{H}| < +\infty$ holds, and then also $|H/N_H| < +\infty$ holds. By hypothesis, $N_H$ satisfies $\mathrm{P}$, and then the thesis follows.
\qed
\begin{corollary}\label{virt-byfin}
A group is finite-by-$\mathbb{Z}^d$ if and only if it is virtually $\mathbb{Z}^d$.
\end{corollary}
\Proof
We have to check that the property of ``being isomorphic to $\mathbb{Z}^d$'' is inherited by subgroups of finite index. By the Fundamental theorem of finitely generated Abelian groups (Theorem~\ref{fundab}), every subgroup $M$ of $N \cong \mathbb{Z}^d$ must be isomorphic to $\mathbb{Z}^{d'}$ with $d' \leq d$, thus $| N / M | = +\infty$ unless $d=d'$.
\qed
It is easy to see that further properties inherited by subgroups of finite index are: cyclicity, Abelianity, freeness.
%
\begin{theorem}[Quasi-isometric rigidity of $\mathbb{Z}^d$~\cite{CTV07}]\label{Zdrigid}
If a finitely generated group $G$ is quasi-isometric to $\mathbb{Z}^d$, then it has a finite index subgroup isomorphic to $\mathbb{Z}^d$.
\end{theorem}
\begin{corollary}\label{cor:Zd-by-finite}
Let $G$ be a finitely generated group. Then $G$ is quasi-isometric to $\mathbb{R}^d$ if and only if $G$ is finite-by-$\mathbb{Z}^d$.
\end{corollary}
\begin{proof}
This straightforwardly follows from the fact that quasi-isometry is an equivalence relation and $\mathbb{R}^d$ is quasi-isometric to $\mathbb{Z}^d$.
\end{proof}
In the light of Corollary~\ref{cor:Zd-by-finite}, our aim is to provide structure results for finite-by-$\mathbb{Z}^d$ groups, along with their presentations, in order to derive admissible (Euclidean) scalar QWs on them. The problem of characterizing the class of groups $G$ with fixed $N \trianglelefteq G$ and quotient $Q=G/N$ is called \emph{group extension problem}: $G$ is indeed said an \emph{extension of $Q$ by $N$}. In Subsec.~\ref{subsec:group_ext} the extension problem will be discussed, while in Subsec.~\ref{subsec:extensionsZd} we shall specialize the analysis to the case $N\cong \mathbb{Z}^d$ and $|Q|<\infty$. Our aim is to provide necessary and sufficient conditions in order to explicitly construct every possible extension of $Q$ by $N$.

\section{The group extension problem}\label{sec:extension_problem}

\subsection{Constructing group extensions}\label{subsec:group_ext}
Let $N,Q$ be two arbitrary groups, and $G$ be a $Q$-by-$N$ group. The cardinality $|Q|$ is called \emph{order of the group $Q$}, and it is precisely the index of $N$ in $G$. The group $G$ can be then partitioned as follows
\begin{align*}
G = \lbrace Nc_{q_1},Nc_{q_2},\ldots , Nc_{q_{|Q|}} \rbrace, 
\end{align*}
where the $c_{q_i}$ are called the \emph{coset representatives}. The identity of $Q$ will be denoted by $\tilde{e}$. One has $c_{\tilde{e}}\coloneqq c_{q_1}\in N$ and $c_{q_i}\not\in N$ for all $i\neq 1$.  The elements $q \in Q$ are in a one-to-one correspondence with the cosets representatives $c_q$, which, by normality of $N$, follow the same composition rule of the elements $q$ up to multiplication by elements of $N$. One has:
\begin{equation}\label{cosets}
c_{q_1}c_{q_2}c_{q_1q_2}^{-1} \in N,\quad \forall q_1,q_2 \in Q .
\end{equation}

By definition, every element $g\in G$ can be written as $g=nc_{q}$, with $n\in N$ and $c_{q}$ the representative of the coset corresponding to $q\in Q$. Then the group multiplication can be obtained as follows
\begin{align*}
n_{1}c_{q_1}n_{2}c_{q_{2}} &= n_{1}c_{q_{1}}n_{2}c^{-1}_{q_{1}}c_{q_{1}}c_{q_{2}}\eqqcolon \\ &\eqqcolon
n_{1}\varphi_{q_{1}}(n_{2})f(q_{1},q_{2})c_{q_{1}q_{2}},
\end{align*}
where
\[\varphi_{q}(n)\coloneqq c_{q}nc_{q}^{-1},\quad f(q_{1},q_{2})\coloneqq c_{q_{1}}c_{q_{2}}c_{q_{1}q_{2}}^{-1},
\]
and clearly $\varphi_{q}\in\mathrm{Aut}(N)$. Therefore, a piece of information we need in order to identify the extension $G$ is the assignment of a composition rule for the coset representatives: this observation motivates the following definition.
\begin{definition}
	\label{2-cocycle}
Let $G$ be an extension of $Q$ by $N$ and $\lbrace c_q\rbrace_{q\in Q}$ a set of representatives of the cosets of $N$ in $G$. The function $f : Q\times Q \rightarrow N$ defined as
\[
f(q_1,q_2) \dfn c_{q_1}c_{q_2}c_{q_1q_2}^{-1},\quad \forall q_1,q_2 \in Q
\]
is called a 2-cocycle.
\end{definition}
%
From relation~\eqref{cosets} one has
\[
\varphi_{q_1} \circ \varphi_{q_2} \circ \varphi_{q_1q_2}^{-1} \in \mathrm{Inn}(N).
\]
Defining
\begin{align*}
\phi_m(\cdot):=m\cdot m^{-1},\quad \forall m\in N,
\end{align*}
one indeed obtains
\begin{align}\label{outer'}
\begin{split}
\varphi_{q_1} \circ \varphi_{q_2} \left( \cdot\right) &= (c_{q_1}c_{q_2}) \cdot (c_{q_1}c_{q_2})^{-1} = \\ &= \phi_{f(q_1,q_2)} \circ \varphi_{q_1q_2} \left(\cdot\right),\quad \forall q_1,q_2\in Q.
\end{split}
\end{align}
Accordingly, the family $\lbrace \varphi_q \rbrace_{q \in Q}$ can be identified, in general, as a family of automorphisms which are in correspondence with elements of the \emph{outer automorphism group} of $N$\footnote{We remind to the reader that an \emph{outer automorphism} is an automorphism which is not inner, namely it does not have a realization as the conjugation by some element of the group. However, the group $\mathrm{Out}(N) \dfn \mathrm{Aut}(N)/\mathrm{Inn}(N)$ is called the \emph{outer automorphism group}, despite the fact that, in general, it does not collect the outer automorphisms of $N$, since it is not generally a subgroup of $\mathrm{Aut}(N)$.}. Since $\mathrm{Inn}(N) \trianglelefteq \mathrm{Aut}(N)$ and $\mathrm{Out}(N) \dfn \mathrm{Aut}(N)/\mathrm{Inn}(N)$, Eq.~\eqref{outer'} induces a homomorphism $\tilde{\varphi}$ which associates an element of $\mathrm{Out}(N)$ to each $q \in Q$. We define:
\begin{align}\label{eq:mapphi}
\begin{split}
\varphi \colon Q &\longrightarrow \mathrm{Aut}(N) \\
q &\longmapsto \varphi_q(\cdot) = c_q\cdot c_q^{-1},
\end{split}
\end{align}
and
\begin{align}\label{eq:projection}
\begin{split}
\pi \colon \mathrm{Aut}(N) &\longrightarrow \mathrm{Out}(N) \\
\nu\circ\xi_{\omega} &\longmapsto \omega,
\end{split}
\end{align}
such that $\nu\in\mathrm{Inn}(N)$, $\xi_{\omega}\in\mathrm{Aut}(N)$ are some coset representative of the cosets of $\mathrm{Inn}(N)$ in $\mathrm{Aut}(N)$, and the composition $\tilde{\varphi}:=\pi\circ\varphi$ is a group homomorphism.

Therefore, in order to identify a group extension $G$ of $Q$ by $N$ (i.e.~in order to give the complete composition rule for the elements of $G$) one needs to choose a family of automorphisms of $N$ and a 2-cocycle, i.e.~a pair $(\varphi,f)$ such that $f$ satisfies Definition~\ref{2-cocycle} and $\varphi$ is defined as in~\eqref{eq:mapphi} and satisfies Eq.~\eqref{outer'}. We call such a pair $(\varphi,f)$ \emph{data} for the extension $G$ of $Q$ by $N$. We are now ready to classify the group extensions in the following Lemma.
\begin{lemma}[Classification of group extensions]\label{l:coc}
	Let $Q$ and $N$ be two groups, and $\varphi,f$ two maps such that $\varphi:Q\to\mathrm{Aut}(N)$, and $f : Q\times Q \rightarrow N$. Then, there exists an extension $G$ of $Q$ by $N$ with data $(\varphi,f)$ if and only if the following relations are satisfied $\forall q_1,q_2,q_3\in Q$:
	\begin{align}
	\varphi_{q_1} \circ \varphi_{q_2} &= \phi_{f(q_1,q_2)} \circ \varphi_{q_1q_2}	,\label{outer} \\
	f(q_1,q_2)f(q_1q_2,q_3) &= \varphi_{q_1}\left(f(q_2,q_3)\right)f(q_1,q_2q_3). \label{b}
	\end{align}
\end{lemma}
\Proof
($\Rightarrow$) Let $G$ be an extension of $Q$ by $N$ with data $(\varphi,f)$. Property~\eqref{outer} has been already shown to hold. It is easy to check property~\eqref{b} for a 2-cocycle by imposing the associativity for the product of coset representatives, namely by the following computation. On the one hand, one has
\begin{align*}
(c_{q_1}c_{q_2})c_{q_3}&= f(q_1,q_2)c_{q_1q_2}c_{q_3} =\\&= f(q_1,q_2)f(q_1q_2,q_3) c_{q_1q_2q_3}
\end{align*}
On the other hand, also
\begin{align*}
c_{q_1}(c_{q_2}c_{q_3})&= c_{q_1}f(q_2,q_3)c_{q_2q_3}=\\&=\varphi_{q_1}\left(f(q_2,q_3)\right)f(q_1,q_2q_3)c_{q_1q_2q_3}
\end{align*}
holds, proving the first implication.

($\Leftarrow$) We now explicitly construct the extension $G$ of $Q$ by $N$ having $(\varphi,f)$ as data.
Let $G'$ be the set of ordered pairs $N \times Q$, and denote its generic element by $g = (n,q)$. Let us also equip $G'$ with the following composition rule:
\begin{align*}
(n_{1},q_{1})(n_{2},q_{2}):=(n_{1}\varphi_{q_{1}}(n_{2})f(q_{1},q_{2}),q_{1}q_{2}).
\end{align*}
Moreover, let the inverse of $g\in G'$ be given by
\begin{align}\label{eq:inverse}
(n,q)^{-1}\coloneqq (\varphi_q^{-1}(n^{-1}f(q,\tilde{e})^{-1})f(q^{-1},q)^{-1},q^{-1}),
\end{align}
and the identity element of $G'$ by $e\coloneqq (f(\tilde{e},\tilde{e})^{-1},\tilde{e})$. 
We now show that $G'$ is isomorphic to an extension of $Q$ by $N$ with data $(\varphi,f)$. 
First, we need to show that $G'$ is actually a group. Associativity of the composition can be proved using the properties~\eqref{outer} and~\eqref{b}. Now, we can prove that $e\in G'$ actually behaves as an the identity as follows.
By property~\eqref{outer} we have
\begin{align}\label{eq:idencocy2}
\varphi_{\tilde{e}} = \phi_{f(\tilde{e},\tilde{e})}.
\end{align}
Moreover, by choosing $q_1=q_2=\tilde{e}$, and using Eq.~\eqref{eq:idencocy2}, relation~\eqref{b} reads
\begin{align}\label{eq:idencocy}
f(\tilde{e},\tilde{e}) = f(\tilde{e},q),\quad \forall q\in Q.
\end{align}
Relation~\eqref{b} with the choice $q_2=q_3=\tilde{e}$ reads
\begin{align}\label{eq:idencocy3}
\varphi_{q}(f(\tilde{e},\tilde{e})) = f(q,\tilde{e}),\quad \forall q\in Q.
\end{align}
Using Eqs.~\eqref{eq:idencocy2}, \eqref{eq:idencocy} and~\eqref{eq:idencocy3}, it is now easy to check that $e(n,q)=(n,q)e=(n,q)$. Eqs.~\eqref{outer} and~\eqref{b} with the choice $q_1=q_3,q_2=q_1^{-1}$ read:
\begin{align}
&\varphi_{q} = \phi_{f(q,q^{-1})}\circ\varphi_{\tilde{e}}\circ\varphi_{q^{-1}}^{-1},\label{eq:idencocy5}\\
&\varphi_{q}(f(q^{-1},q)^{-1})f(q,q^{-1}) = f(q,\tilde{e})f(\tilde{e},q)^{-1}.\label{eq:idencocy4}
\end{align}
To prove that the inverse of $(n,q)$ is well defined by relation~\eqref{eq:inverse}, one has to use~\eqref{eq:idencocy} and~\eqref{eq:idencocy4} for the right-multiplication, and relations~\eqref{eq:idencocy3} and~\eqref{eq:idencocy5} for the left-multiplication. We can now show that the group $G'$ is $Q$-by-$N$. Let us define $r_n\coloneqq (nf(\tilde{e},\tilde{e})^{-1},\tilde{e})$. The subset $N'\coloneqq \lbrace r_n \left.|\right. n\in N \rbrace$ forms a subgroup of $G'$, thus $r_mr_nr_m^{-1}\in N'$. Let now define $c_q\coloneqq (e_N,q)$. One has:
\begin{align}\label{eq:automorph}
\begin{split}
c_qr_nc_q^{-1}= (\varphi_q(n)f(\tilde{e},\tilde{e})^{-1},\tilde{e}) \equiv r_{\varphi_{q}(n)},
\end{split}
\end{align}
where we used Eqs.~\eqref{eq:idencocy}, \eqref{eq:idencocy3} and~\eqref{eq:idencocy4}. Thus, since the general element $(n,q)\in G'$ can be expressed as $r_nc_q$, Eq.~\eqref{eq:automorph} shows that the subgroup $N'$ is normal in $G'$. Moreover, we have that
\begin{align}\label{eq:cocyle}
c_{q_1}c_{q_2} = (f(q_1,q_2)f(\tilde{e},\tilde{e})^{-1},\tilde{e})c_{q_1q_2}\equiv r_{f(q_1,q_2)}c_{q_1q_2}.
\end{align}
Moreover, one verifies that $N'$ is indeed isomorphic to $N$:
\begin{align*}
r_{n_1}r_{n_2} &= (n_1f(\tilde{e},\tilde{e})^{-1},\tilde{e})(n_2f(\tilde{e},\tilde{e})^{-1},\tilde{e}) =\\ &= (n_1n_2f(\tilde{e},\tilde{e})^{-1},\tilde{e}) = r_{n_1n_2}.
\end{align*}
On the other hand, the quotient $G'/N'$ is clearly isomorphic to $Q$. We have then proven that $G'$ is $Q$-by-$N$. Finally, we can now define an isomorphism $\zeta:G'\to G$ by setting $\zeta(r_n)=n$ and $\zeta(c_q)=c_q$, where $N$ is now a normal subgroup of $G$ with quotient $G/N\cong Q$. Thus, $G$ is an extension of $Q$ by $N$. Moreover, since relation~\eqref{outer} holds, and by Eqs.~\eqref{eq:automorph} and~\eqref{eq:cocyle} one has 
\begin{align*}
&\zeta(c_qr_n c_q^{-1})=c_q n c_q^{-1}=\zeta(r_{\varphi_q(n)})=\varphi_q(n),\\
&\zeta(c_{q_1}c_{q_2}) =c_{q_1}c_{q_2}=\zeta(r_{f(q_1,q_2)}c_{q_1q_2})=f(q_1,q_2)c_{q_1q_2},
\end{align*}
then the group $G$ has data $(\varphi, f)$.
\qed
Lemma~\ref{l:coc} extends the result proven in Ref.~\cite{R12} for the case where $N$ is Abelian. By Lemma~\ref{l:coc}, in order to construct and classify the extensions of $Q$ by $N$ one has to choose: (i) a map $\varphi$ defined as in~\eqref{eq:mapphi}, and (ii) a map $f : Q\times Q \rightarrow N$, such that properties~\eqref{outer} and~\eqref{b} are satisfied. In particular, in the case where $N$ is Abelian, by Eq.~\eqref{outer} one has that the map $\varphi:Q\to\mathrm{Aut}(N)$ is a group homomorphism, since in this case $\mathrm{Inn}(N)$ is trivial.

Yet, in general two extensions with different choices of data $(\varphi,f)$ may still be isomorphic. We shall now prove a sufficient conditions implying that two extensions having different data are indeed isomorphic.
\begin{definition}
Let $G,G'$ be two extensions of $Q$ by $N$ and of of $Q'$ by $N'$, respectively. $G$ and $G'$ are called pseudo-congruent extensions if there exist: (i) an isomorphism $\psi : G '\rightarrow G$, (ii) two isomorphisms $\alpha : N' \rightarrow N$ and $\beta: Q' \rightarrow Q$, and (iii) a family $\lbrace n_q\left.|\right. q\in Q' ,n_q\in N' \rbrace$, such that
\begin{equation}\label{pseudo}
\psi(n)=\alpha(n),\quad \psi(c'_{q})=\alpha(n_q)c_{\beta(q)},\quad \forall n \in N',\forall q\in Q'.
\end{equation}

\end{definition}
\begin{lemma}[Classification of pseudo-congruent extensions]\label{pseudolemma}
Let $N,N',Q,Q'$ be groups such that $N\cong N',Q\cong Q'$, and $G'$ an extension of $Q'$ by $N'$ with data $(\varphi',f')$. Then, there exists an extension $G$ of $Q$ by $N$ with data $(\varphi,f)$ and pseudo-congruent to $G'$, if and only if there exist two isomorphisms $\alpha : N' \rightarrow N$ and $\beta: Q' \rightarrow Q$, and a family $\lbrace n_q\left.|\right. q\in Q' ,n_q\in N' \rbrace$, such that
\begin{align}
\begin{split}\label{pseudo2}
&\varphi_{\beta(q_1)} = \phi_{\alpha(n_{q_1}^{-1})} \circ \alpha \circ  \varphi'_{q_1} \circ \alpha^{-1},
\end{split}\\
\begin{split}\label{cocycond1}
&f(\beta(q_1),\beta(q_2))=\\&=\varphi_{\beta(q_1)}\circ\alpha(n_{q_2}^{-1})\alpha(n_{q_1}^{-1}f'(q_1,q_2)n_{q_1q_2}),\quad \forall q_1,q_2\in Q.
\end{split}
\end{align}
\end{lemma}
\Proof
($\Rightarrow$) Let $N,N',Q,Q'$ be groups such that $N\cong N'$ and $Q\cong Q'$. Let then $G',G$ be two extensions of $Q'$ by $N'$ and of $Q$ by $N$ with data $(\varphi',f')$ and $(\varphi,f)$, respectively. By hypothesis, there exists an isomorphism $\psi : G' \rightarrow G$, two isomorphisms $\alpha : N' \rightarrow N$ and $\beta: Q' \rightarrow Q$, and a family $\lbrace n_q\left.|\right. q\in Q',n_q\in N' \rbrace$, satisfying the following:
\begin{align*}
\psi(n)=\alpha(n),\quad \psi(c'_{q})=\alpha(n_q)c_{\beta(q)},\quad \forall n \in N',\forall q\in Q'.
\end{align*}
For all $n\in N'$ and $q_1,q_2\in Q'$, one has $c_{q_1}'nc_{q_2}' = \varphi'_{q_1}(n)f'(q_1,q_2)c_{q_1q_2}'$. Then, letting $\psi$ act on both sides of the latter relation, it follows, $\forall n\in N', \forall q\in Q'$, that
\begin{align}\label{isocond}
\begin{split}
&\alpha(n_{q_1})c_{\beta(q_1)}\alpha(nn_{q_2})c_{\beta(q_2)} = \\ &= \alpha(n_{q_1})\varphi_{\beta(q_1)} \circ \alpha (nn_{q_2})f(\beta(q_1),\beta(q_2))c_{\beta(q_1q_2)} = \\&=\alpha \circ \varphi'_{q_1} (n )\alpha (f'(q_1,q_2)n_{q_1q_2})c_{\beta(q_1q_2)}.
\end{split}
\end{align}
Choosing $n=e$ in Eq.~\eqref{isocond}, one obtains:
\begin{align}\label{cocycond}
\begin{split}
&f(\beta(q_1),\beta(q_2))=\\&=\varphi_{\beta(q_1)}\circ\alpha(n_{q_2}^{-1})\alpha(n_{q_1}^{-1})\alpha(f'(q_1,q_2)n_{q_1q_2}).
\end{split}
\end{align}
Using condition~\eqref{cocycond}, Eq.~\eqref{isocond} finally gives relation~\eqref{pseudo2}.

($\Leftarrow$) By hypothesis, $G'$ is an extension of $Q'$ by $N'$ with data $(\varphi',f')$. Moreover, there exist two groups $N,Q$, a map $\varphi \colon Q \to \mathrm{Aut}(N)$, two isomorphisms $\alpha : N' \rightarrow N$ and $\beta: Q' \rightarrow Q$, and a family $\lbrace n_q\left.|\right. q\in Q ,n_q\in N \rbrace$, satisfying Eqs.~\eqref{pseudo2},~\eqref{cocycond1}. Then, using relations~\eqref{outer} and~~\eqref{pseudo2} for the data $(\varphi',f')$, we have
\begin{align}
\begin{split}\label{eq:cocy'}
&\phi_{f'(q_1,q_2)}\circ\varphi'_{q_1q_2}=
\varphi'_{q_1}\circ\varphi'_{q_2}=\\&=\phi_{n_{q_1}}\circ \alpha^{-1}\circ\varphi_{\beta(q_1)}\circ\alpha\circ\phi_{n_{q_2}}\circ \alpha^{-1}\circ\varphi_{\beta(q_2)}\circ\alpha=\\&=
\alpha^{-1}\circ\phi_{\alpha(n_{q_1})}\circ\varphi_{\beta(q_1)}\circ\phi_{\alpha(n_{q_2})}\circ\varphi_{\beta(q_2)}\circ\alpha=\\&=
\alpha^{-1}\circ\phi_{\alpha(n_{q_1})\varphi_{\beta(q_1)}\circ\alpha(n_{q_2})}\circ\varphi_{\beta(q_1)}\circ\varphi_{\beta(q_2)}\circ\alpha,\\
&\forall q_1,q_2\in Q
\end{split}
\end{align}
that reads:
\begin{align}\label{eq:cocyphialpha}
\begin{split}
&\varphi_{\beta(q_1)}\circ\varphi_{\beta(q_2)}=
\\&= \phi_{\varphi_{\beta(q_1)}\circ\alpha(n_{q_2}^{-1})\alpha(n_{q_1}^{-1})}\circ\alpha\circ\phi_{f'(q_1,q_2)}\circ\varphi'_{q_1q_2}\circ\alpha^{-1} =
\\&= \phi_{\varphi_{\beta(q_1)}\circ\alpha(n_{q_2}^{-1})\alpha(n_{q_1}^{-1})\alpha(f'(q_1,q_2))}\circ\alpha\circ\varphi'_{q_1q_2}\circ\alpha^{-1} =\\&=
\phi_{\varphi_{\beta(q_1)}\circ\alpha(n_{q_2}^{-1})\alpha(n_{q_1}^{-1}f'(q_1,q_2)n_{q_1q_2})}\circ\varphi_{\beta(q_1q_2)}=\\&=
\phi_{f(\beta(q_1),\beta(q_2))}\circ\varphi_{\beta(q_1q_2)},\quad \forall q_1,q_2\in Q.
\end{split}
\end{align}
Therefore, by Eq.~\eqref{eq:cocyphialpha}, the maps $\varphi,f$ satisfy property~\eqref{outer}. Moreover, expressing $f'$ in terms of $\varphi',f$ by Eqs.~\eqref{pseudo2},\eqref{cocycond1} and writing property~\eqref{b}, it is straightforward to verify, using Eq.~\eqref{eq:cocy'}, that property~\eqref{b} also holds for $\varphi,f$. Thus, Lemma~\eqref{l:coc} guarantees the existence of an extension $G$ of $Q$ by $N$ with data $(\varphi,f)$. Let $\psi : G' \rightarrow G$ be map defined by: $\psi(nc_{q}')=\alpha(n)\alpha(n_q)c_{\beta(q)}$ $\forall n \in N',\forall q\in Q'$. On the one hand, one has:
\begin{align*}
&\psi\left( n_1c'_{q_1}n_2c'_{q_2}\right)=\\&= 
\psi\left( n_1\varphi_{q_1}(n_2)f'(q_1,q_2)c'_{q_1q_2}\right) =\\&= \alpha(n_1)\alpha\circ\varphi'_{q_1}(n_2)\alpha(f'(q_1,q_2)n_{q_1q_2})c_{\beta(q_1q_2)}.
\end{align*}
On the other hand, also
\begin{align*}
&\psi\left( n_1c'_{q_1}\right)\psi\left( n_2c'_{q_2}\right) = \alpha(n_1n_{q_1})c_{\beta(q_1)}\alpha(n_2n_{q_2})c_{\beta(q_2)} =\\&= \alpha(n_1n_{q_1})\varphi_{\beta(q_1)}\circ\alpha(n_2n_{q_2})f(\beta(q_1),\beta(q_2))c_{\beta(q_1q_2)}
\end{align*}
holds. Using relation~\eqref{pseudo2} and Eq.~\eqref{cocycond1}, one finally concludes that $\psi\left( n_1c'_{q_1}n_2c'_{q_2}\right)=\psi\left( n_1c'_{q_1}\right)\psi\left( n_2c'_{q_2}\right)$ $\forall n_1,n_2\in N'$, $q_1,q_2\in Q'$. Consequently, $\psi$ is an isomorphism and then the two extensions $G'$ and $G$ are pseudo-congruent.
\qed
\begin{corollary}\label{c:fe}
One can always choose
\begin{align}\label{eq:fe}
c_{\tilde{e}}=f(\tilde{e},\tilde{e})=f(\tilde{e},q)=f(q,\tilde{e})=e,\quad \forall q\in Q
\end{align}
up to pseudo-congruence.
\end{corollary}
\Proof
We have already proven that $f'(\tilde{e},\tilde{e})=f'(\tilde{e},q)$ $\forall q\in Q$ (see Eq.~\eqref{eq:idencocy}). Then, by choosing $\alpha,\beta$ to be the identical maps, $n_{\tilde{e}}=f'(\tilde{e},\tilde{e})$, $q_1=\tilde{e}$, $q_2=q$, and using Eq.~\eqref{eq:idencocy2}, relation~\eqref{cocycond1} reads: $f(\tilde{e},\tilde{e})=f(\tilde{e},q)=e$ $\forall q\in Q$. Moreover, by Eq.~\eqref{eq:idencocy3}, one has $e=\varphi_{q}(f(\tilde{e},\tilde{e})) = f(q,\tilde{e})$ $\forall q\in Q$.
\qed
Lemma~\ref{l:coc} and Lemma~\ref{pseudolemma}, along with Corollary~\ref{c:fe}, imply the following result.
\begin{proposition}[Construction of group extensions up to pseudo-congruence]\label{prop:construction}
Let $N$ and $Q$ be two groups, and $\pi$ the projection map in~\eqref{eq:projection}.
The following procedure allows one to explicitly construct all the extensions of $Q$ by $N$ up to pseudo-congruence:
\begin{enumerate}
\item Classify all the homomorphisms $\tilde{\varphi}:Q\to\mathrm{Out}(N)$ up to the equivalence relation defined as follows:
\begin{align}\label{eq:equiv1}
\begin{split}
&\tilde{\varphi}^{(1)} \sim \tilde{\varphi}^{(2)}\Leftrightarrow\\ 
& \quad\exists~\alpha\in\mathrm{Aut(N)},~\beta\in\mathrm{Aut(Q)}, ~\varphi:Q\to \mathrm{Aut}(N)~\colon\\
&\tilde{\varphi}^{(1)}_q = \pi\circ\alpha^{-1}\circ\varphi_{\beta(q)}\circ\alpha,\quad
\tilde{\varphi}^{(2)}_q = \pi\circ\varphi_{q},\quad\forall q\in Q.
\end{split}
\end{align}
For each chosen homomorphism $\tilde{\varphi}^*$ (modulo the equivalence defined above), choose a map $\varphi^*$ such that $\tilde{\varphi}^*=\pi\circ\varphi^*$. 
\item Classify all the maps $f:Q\times Q\to N$ satisfying properties~\eqref{outer} and~\eqref{b} with $\varphi=\varphi^*$, and~\eqref{eq:fe} up to the following equivalence relation:
\begin{align}\label{changecoc2}
\begin{split}
&f^{(1)}\sim f^{(2)}\ \Leftrightarrow\ \exists \lbrace n_q\left.|\right. q\in Q ,n_q\in N \rbrace,\\ 
&\quad\exists~\alpha\in\mathrm{Aut(N)},~\beta\in\mathrm{Aut(Q)},~\vartheta:Q\to \mathrm{Aut}(N)~\colon\\
&f^{(1)}(\beta(q_1),\beta(q_2))=\\&=\vartheta_{\beta(q_1)}\circ\alpha(n_{q_2}^{-1})\alpha(n_{q_1}^{-1}f^{(2)}(q_1,q_2)n_{q_1q_2}),~\forall q_1,q_2\in Q.
\end{split}
\end{align}
Choose a map $f^*$ (modulo the equivalence defined above). 
\end{enumerate}
Each pair of maps $\varphi^*,f^*$ obtained in this way provides, up to pseudo-congruence, a different extension $G$ of $Q$ by $N$, whose data are $(\varphi^*,f^*)$.
\end{proposition}
Although a general criterion allowing to classify and construct extensions up to an arbitrary isomorphism is not known, the structure of the particular groups $N,Q$ under study can possibly help to recognize whether or not two extensions are isomorphic, while being not pseudo-congruent. Such a structure can also help to explicitly construct the equivalence classes defined in relations~\eqref{eq:equiv1},~\eqref{changecoc2} of Proposition~\ref{prop:construction}, and we give an example of this in the next subsection.

\subsection{Finite-by-$\mathbb{Z}^d$ extensions}\label{subsec:extensionsZd}
Our aim is to study the Euclidean scalar QWs, namely walks on Cayley graphs quasi-isometric to $\mathbb{R}^d$. By Corollary~\ref{cor:Zd-by-finite}, this can be accomplished without loss of generality constructing the extensions of some finite group $Q$ by $N\cong\mathbb{Z}^d$. Notice that, in the case where the group $N$ is Abelian, then one has $\mathrm{Inn}(N)=\lbrace\mathrm{id}\rbrace$, namely $\mathrm{Aut}(N) \cong \mathrm{Out}(N)$. Consequently, being the projection map $\pi$ trivial, we shall set $\tilde{\varphi}\equiv\varphi$.
\begin{theorem}[\cite{DK17}]\label{thm:GLdZ}
The group of automorphisms of $\mathbb{Z}^d$ is isomorphic to $\mathbb{GL}(d,\mathbb{Z})$.
\end{theorem}
By relation~\eqref{eq:equiv1} in Proposition~\ref{prop:construction} combined with Theorem~\ref{thm:GLdZ}, for each fixed finite group $Q$, one can consider the equivalence classes of maps $\varphi:Q\to\mathbb{GL}(d,\mathbb{Z})$ up to pre-composition with  arbitrary $\beta \in \mathrm{Aut}(Q)$ and to conjugation by arbitrary $\alpha\in\mathbb{GL}(d,\mathbb{Z})$.
\begin{lemma}\label{ord}
	Every element $q$ of a finite group $Q$ has a finite order $r_{q}$.
\end{lemma}
%
Lemma~\ref{ord}, using property~\eqref{outer} along with the fact that $\mathrm{Inn}(N)$ is trivial, implies that
\[
\forall q \in Q\ \exists r_q\in \mathbb{N}\;\colon \varphi_q^{r_q} (n) = n,\quad \forall n\in N.
\] 
Such an automorphism $\varphi_q$ is called $r_q$-\emph{involutory}. Since $\varphi$ is a group homomorphism, we need the conjugacy classes of finite subgroups of $\mathbb{GL}(d,\mathbb{Z})$, whose elements will be $r_q$-involutory matrices. This fact will also help in finding the equivalence classes defined in Eq.~\eqref{changecoc2} in Proposition~\ref{prop:construction}.

In the following, we shall focus our attention in particular on the extensions by $\mathbb{Z}^d$ with index $|Q| = 2$, since it allows to reconstruct QWs with a two-dimensional coin. In this case, by Corollary~\ref{c:fe}, the problem reduces to choose the only nontrivial value of the 2-cocycle $f$, namely $f(\tilde{q},\tilde{q}) = c_{\tilde{q}}^2 \eqqcolon c^2$, which by property~\eqref{b} is invariant under the automorphism $\varphi_{\tilde{q}}$. Accordingly, one can calculate the invariant space of $\varphi_{\tilde{q}}$ and then use Eq.~\eqref{changecoc2} to find all the inequivalent 2-cocycles. We notice that Eq.~\eqref{cocycond1} and Corollary~\ref{c:fe} imply $n_{\tilde{e}}=e$. Furthermore, for $Q \cong \mathbb{Z}_2$, $\mathrm{Aut}(Q)=\lbrace\mathrm{id}\rbrace$. Therefore in this case, for a chosen ${c}^2$, by Eq.~\eqref{changecoc2} one looks for solutions $n,c'^2\in N\cong \mathbb{Z}^d$ to the following equation:
\begin{equation}\label{cosetchange2}
{c}^2 = n + \varphi_{\tilde{q}}(n) + {c'}^2
\end{equation}
(where we used the Abelian additive notation). We will neglect the trivial homomorphism $\varphi$ (from $Q \cong \mathbb{Z}_2$ to $I_d$), since it gives rise to Abelian extensions, implying a trivial scalar QW dynamics~\cite{BDEPT16}. The explicit construction of all the extensions for the cases $d=1,2,3$ and $Q\cong \mathbb{Z}_2$ will be performed in Sec.~\ref{sec:relativistic_equations}, along with two extensions with $Q\cong\mathbb{Z}_2\times \mathbb{Z}_2$.

\subsection{Coarse-graining of QWs on Cayley graphs}\label{subsec:coarse-graining}

Let $G$ be a finitely generated group and $N$ a subgroup of $G$. One can define a unitary mapping between $\mathscr{H}=\ell^2(G)$ and $\mathscr{K}=\ell^2(N)\otimes \ell^2(G/N)$ in the following way:
\begin{align*}
U_N \colon \mathscr{H}\ &\longrightarrow \ \mathscr{K},\\
     \ket{nc_q}\ &\longmapsto \ \kket{n}{q},
\end{align*}
We notice that $N$ has not to be necessarily normal in $G$. For all $n\in N$, $h\in G$ and $q\in G/N$, there exist $n'\in N$ and $q'\in G/N$ such that $nc_q h^{-1} = n'c_{q'(h,q)}$. In particular, $n'=n (c_{q'(h,q)}hc_q^{-1})^{-1}$, and $c_{q'(h,q)}hc_q^{-1}\in N$ for all $q\in G/N$. One can then provide a representation for $G$ in terms of the right-regular representation of $N$:
\begin{align}\label{eq:renorm-gen}
\begin{split}
\tilde{T}_h & = U_N T_h U_N^{\dagger} =
\sum\limits_{q\in G/N} T_{c_{q'(h,q)}hc_{q}^{-1}} \otimes \ket{q'(h,q)}\! \bra{q}.
\end{split}
\end{align}

The abovementioned change of representation has been used in Ref.~\cite{DEPT16} to define the so-called \emph{coarse-graining of QWs}. Suppose to have a finitely generated group $G$ with a finite-index subgroup $N$. In the following lemma, we prove that a QW on the Cayley graph $\Gamma(G,S_+)$ with a $s$-dimensional coin can be represented as a QW on the Cayley graph of $N$ having the following set of generators 
\begin{equation}\label{ren-gen}
S_{N} \coloneqq \{ c_{q'(h,q)}hc_{q}^{-1} \}_{h\in S}^{q\in G/N},
\end{equation}
and with an enlarged coin of dimension $s\cdot |G/N|$.
\begin{lemma}[Characterization of coarse-grainings of QWs]\label{gen-N}
Let $G$ be a group, $S$ a set of generators for $G$ and $N$ a subgroup of $G$. Let the coset representatives of $N$ in $G$ be denoted by $\lbrace c_q\rbrace_{q\in G/N}$ with $c_{q_1}\in N$. Then 
$S_N$ is a set of generators for $N$.
\end{lemma}
\begin{proof}
Let $\tilde{S}_N$ be a set of generators for $N\leq G$. For all $\tilde{h}\in \tilde{S}_N$, by hypothesis there exist $h_{1},h_{2},\ldots ,h_{n}\in S$ such that
\[
\tilde{h} = h_{1}h_{2}\cdots h_{n}.
\]
Using elements in the set $S_N$, we can now recursively construct all the elements of the form:
\begin{align*}
& c_{{q'(h_1},q'(q'(...))}\tilde{h}c_{q_1}^{-1} \coloneqq \\ &
c_{{q'(h_1},q'(q'(...))}h_{1} \cdots  h_{{n-1}}c_{q'(h_n,1)}^{-1} c_{q'(h_n,1)}h_{n}c_{q_1}^{-1},
\end{align*}
which are in $N$ by construction (since $c_{q'(h,q)}hc_q^{-1}\in N$ for all $h\in G$ and $q\in G/N$). Then it must be $c_{q'(h_{1},q'(q'(...))}\in N$, meaning that $q'(h_{1},q'(q'(...)) = q_1$. We thus constructed $c_{q_1} \tilde{h} c_{q_1}^{-1}$ for all generators $\tilde{h}$ of $N$. Hence the thesis follows noticing that the set $\lbrace c_{q_1}\tilde{h}c_{q_1}^{-1}\in N \left.|\right. \tilde{h}\in\tilde{S}_N\rbrace$ generates the whole $N$.
\end{proof}
In Lemma~\ref{gen-N} we have proven that the coarse-graining of QWs on Cayley graphs preserves the whole subgroup $N\leq G$, namely the coarse-grained walk is rigorously defined on a Cayley graph of $N$. Moreover, being a unitary mapping, this coarse-graining does not erase  information of the original walk: some degrees of freedom are just encoded in the new enlarged coin. This circumstance was verified for particular cases in Ref.~\cite{DEPT16}. Notice that the coarse-graining is well defined with no need to assume the normality of $N$ in $G$.

Let now $G$ be an extension of $Q$ by $N$. For an arbitrary $h = x_hc_{q_h}\in G$, expression~\eqref{eq:renorm-gen} reads
\begin{align*}
\tilde{T}_h &= U_N T_h U_N^{\dagger} = \\ &= \sum\limits_{q_1,q_2\in Q} \sum\limits_{n,n'\in N} \ket{n} \! \bra{n'} \otimes \ket{q_1}\! \bra{q_2} \delta_{nc_{q_1},n'c_{q_2}c_{q_h}^{-1}x_h^{-1}}.
\end{align*}
Accordingly, one obtains 
\begin{align*}
q_2 = q_1q_h,\quad n= n'c_{q_1q_h}c_{q_h}^{-1}c_{q_1}^{-1}\varphi_{q_1}(x_h^{-1}),
\end{align*}
and the following holds:
\begin{align*}
\tilde{T}_h =
\sum\limits_{q'\in Q} T_{\varphi_{q'}(x_h)} T_{f(q',q_h)} \otimes \ket{q'}\! \bra{q'} T_{q_h}.
\end{align*}
One can thus give the following representation of the scalar QW evolution operator~\eqref{walkop} on $\Gamma(G,S_+)$ in terms of the right-regular representations of $N$ and $Q$:
\begin{align}\label{walkopnew}
\tilde{A} = \sum_{h\in S} \sum_{q'\in Q}  T_{\varphi_{q'}(x_h)} T_{f(q',q_h)} \otimes z_{h}\ket{q'}\! \bra{q'} T_{q_h},
\end{align}
where $z_{h} \in \Cmplx$ and $h = x_hc_{q_h}\in G$. Finally, the transition matrices corresponding to the coarse-grained generators $\tilde{h}\in S_N$ are given by (see Ref.~\cite{BDEPT16}):
\begin{align}\label{eq:coarsed_trans_matr}
(A_{\tilde{h}})_{ij} = \sum_{h\in S} \delta_{\tilde{h},c_{q_i} hc_{q_j}^{-1}}\delta_{q_i,q_jq_h^{-1}} z_h.
\end{align}
Notice that each transition scalar is associated to one and only one coarse-grained matrix element. Finally, via the above construction we have proven the following.
\begin{proposition}
Let $\mathrm{P}$ be a group property. The set of QWs on Cayley graphs of virtually $\mathrm{P}$ groups is contained in the set of QWs on Cayley graphs of groups satisfying the property $\mathrm{P}$.
\end{proposition}
\begin{corollary}
The Euclidean QWs are contained in the set of QWs on $\mathbb{Z}^d$. The Euclidean coinless QWs are contained in the set of QWs on $\mathbb{Z}^d$ with a $s$-dimensional coin, for $s\geq 2$.
\end{corollary}

\section{Euclidean scalar quantum walks}\label{sec:euclidean_scalar_QWs}
In the following, when we consider the elements of $N\cong \mathbb{Z}^d$ as vectors embedded in $\mathbb{R}^d$, we use the boldfaced notation, indicating by $\v{n}$ the vector corresponding to the element $n\in N\cong \mathbb{Z}^d$. Depending on the context and without loss of clarity, in the index-2 case where $Q = \lbrace \tilde{e}, \tilde{q} \rbrace \cong \mathbb{Z}_2$, we  adopt the identification $\varphi_{\tilde{q}}\equiv \varphi$.

\subsection{Isotropic scalar QWs}\label{subsec:isotropic_scalar_QWs}
In the present subsection, we treat the isotropic case. We start by proving some useful results.

\begin{proposition}\label{prop:char}
Let $N\cong \mathbb{Z}^d$ be an index-2 subgroup of a non-Abelian group $G=\lbrace N,Nc\rbrace$. Then $N$ is characteristic in $G$, namely $N$ is invariant under the action of every automorphism $l\in \mathrm{Aut}(G)$. In particular, $N\trianglelefteq G$.
\end{proposition}
\begin{proof}
By contradiction, suppose that
\begin{align*}
\exists l\in \mathrm{Aut}(G),\exists n\in N \colon l(n)=n'c.
\end{align*}
Accordingly,
\begin{align}\label{eq:zentrum}
l(n^2) = n'\varphi(n')c^2\in Z(G)
\end{align}
but $n\not\in Z(G)$, since every automorphism $l$ maps $Z(G)$ to itself and $l(n)=n'c\not\in Z(G)$. On the other hand, for all $r\in \mathbb{N}^+$ and $g\in N$, one has that
\begin{align*}
g^r\in Z(G)\ \Longrightarrow\ \varphi(g)^rg^{-r}=e,
\end{align*}
implying $g\in Z(G)$, since $N$ is free Abelian. However, by Eq.~\eqref{eq:zentrum}, it must be $n^2\in Z(G)$ and then also $n\in Z(G)$, which is absurd. Then $N$ is characteristic in $G$.
\end{proof}
Using Proposition~\ref{prop:char} we can now construct the automorphism group of a $\mathbb{Z}_2$-by-$\mathbb{Z}^d$ group.
\begin{proposition}\label{prop:autom}
Let $G$ be a $\mathbb{Z}_2$-by-$\mathbb{Z}^d$ group. Then $\mathrm{Aut}(G) \cong \mathbb{Z}^d \rtimes \mathbb{GL}(d,\mathbb{Z})$.
\end{proposition}
\begin{proof}
Let us pose $G=\lbrace N,Nc \rbrace$ and $N\cong \mathbb{Z}^d$. By Proposition~\ref{prop:char}, any automorphism of $G$ acts separately on the two cosets. Let $L_N,L_c$ be the proper subgroups of $\mathrm{Aut}(G)$ such that, for all $l_N\in L_N,l_c\in L_c, n\in N$, then
\begin{align*}
l_N(n) &\in N,\quad l_N(c) = c \\
l_c(n) &= n,\quad l_c(c)   = n_c c,\quad n_c\in N.
\end{align*}
Then, for all $l\in \mathrm{Aut}(G)$ and for all $n\in N$, one can define the two automorphisms $l_N\in L_N,l_c\in L_c$ by
\begin{align*}
l_N(n)  \coloneqq l(n),\quad l_c(c) \coloneqq l(c),
\end{align*}
so that $l  =l_N\circ l_c $.
Furthermore, $L_c\cap L_N = \lbrace e \rbrace$ clearly holds. Moreover, for all $ l_N\in L_N,l_c\in L_c, n\in N$, the following hold:
\begin{align*}
\l_N\circ \l_c\circ l_N^{-1}(n)  &= n, \\
\l_N\circ \l_c\circ l_N^{-1}(nc) &= nl_N(l_c(c)) = nl_N(n_c)c \eqqcolon l_c'(nc),
\end{align*}
where $l_c'\in L_c$. This means that $L_c\trianglelefteq \mathrm{Aut}(G)$. We conclude that $\mathrm{Aut}(G)\cong L_c\rtimes L_N$. Finally, $L_N\cong \mathrm{Aut}(N)\cong \mathbb{GL}(d,\mathbb{Z})$ (by Theorem~\ref{thm:GLdZ}), while $L_c\cong \mathbb{Z}^d$: it is easy to verify that $L_c$ is free Abelian and the cardinality of a minimal set of generators is $d$ (the elements of $L_c$ are the $d$-dimensional translations).
\end{proof}
In Ref.~\cite{d2017isotropic} it is proven that the isotropy group $L$ (Definition~\ref{isotropy}) of a QW must be indeed a finite subgroup of $\mathrm{Aut}(G)$. Moreover, a technique for constructing Cayley graphs starting from $L$ is presented. For an isotropic QW, the generating set coincides with the orbit of an element $h$ under the isotropy group $L$, denoted by $\mathcal{O}_L(h)$. Our aim is now to characterize the isotropic presentations of $\mathbb{Z}_2$-by-$\mathbb{Z}^d$ groups $G$, in order to investigate  the admissible isotropic scalar QWs on them.

By Proposition~\ref{prop:autom}, we know that the isotropy groups of a $\mathbb{Z}_2$-by-$\mathbb{Z}^d$ group are isomorphic to the finite subgroups of $L_N\cong \mathbb{GL}(d,\mathbb{Z})$, since $L_c\cong \mathbb{Z}^d$ has no nontrivial finite subgroups. The finite groups $L\leq \mathbb{GL}(d,\mathbb{Z})$ are reported in Ref.~\cite{d2017isotropic}. In particular, as long as the (isotropic) Cayley graphs of $\mathbb{Z}^d$ are embedded in $\mathbb{R}^d$, the elements in the orbit $\mathcal{O}_L(h)$ can be represented having all the same length~\cite{d2017isotropic}. We shall use the above results to prove the following no-go result.

\begin{proposition}[No-go for Euclidean isotropic scalar QWs]\label{prop:no-go}
In dimensions $d\geq 2$, there exists no isotropic scalar QW on a $\mathbb{Z}_2$-by-$\mathbb{Z}^d$ non-Abelian group $G$. In other words, for $d\geq 2$ no QW on $\mathbb{Z}^d$ with a two-dimensional coin can be derived from an isotropic scalar QW by coarse-graining. For $d=1$, there exists a family of QWs on the infinite dihedral group $D_\infty=\mathbb{Z}\rtimes_{\varphi}\mathbb{Z}_2$, with $\varphi(n)=-n$, containing the 1D Weyl and Dirac QWs.
\end{proposition}
\begin{proof}
Let us suppose that $S = \{ (g_lc)^{\pm 1} \left.|\right. 2\leq l < \infty \}$ is a set of generators for the non-Abelian group $G=\lbrace N,Nc \rbrace$, with $N\cong \mathbb{Z}^d$ and $d\in \mathbb{N}^+$. Take $g_{l_1},g_{l_2}$ such that the length-2 path $g_{l_1}c (g_{l_2}c)^{-1} = g_{l_1}g_{l_2}^{-1}$, for $g_{l_1}c,g_{l_2}c\in S$, is of maximal length in $N$ among all paths of the form $g_{l}g_{l'}^{-1}$. Following the argument of the proof of Proposition 1 in Ref.~\cite{BDEPT16}, the pair $(g_{l_1},g_{l_2})$ is unique in $N$. Accordingly, by the quadrangularity condition (Proposition~\ref{quadrang}), we have to include $g_ic,g_jc\in Nc$ in $S$ such that
\begin{align*}
g_{l_1}g_{l_2}^{-1} = g_i g_j^{-1}
\end{align*}
and with $g_i = g_{l_2}^{-1}$ and $g_j = g_{l_1}^{-1}$. This implies that $g_{l_1}^{\pm 1}c\in S$. Imposing isotropy, and using Proposition~\ref{prop:char} and~\ref{prop:autom} along with aforementioned characterization of the isotropic presentations of Ref.~\cite{d2017isotropic}, we see that all the $g_l$ are all equal in length. Accordingly, $g_{l_1}c(g_{l_1}^{-1}c)^{-1} = g_{l_1}^2$ has maximal length in $N$, and there cannot exist a different pair $(g_{i'}c,g_{j'}c)$ such that $g_{i'}c(g_{j'}c)^{-1} = g_{l_1}^2$. Then, by quadrangularity (see Proposition~\ref{quadrang}), the set of generators for the group $G$ must have the form $S = \{ g_n^{\pm 1} , (g_mc)^{\pm 1} \left.|\right. n\in I, m\in J, |I|\geq 1,|J| \geq 1 \}$, with $\lVert \v{g}_{n_1}\rVert = \lVert \v{g}_{n_2}\rVert$ and $\lVert \v{g}_{m_1}\rVert = \lVert \v{g}_{m_2}\rVert$ for all $n_1,n_2\in I$ and $m_1,m_2\in J$. Moreover, by isotropy, the $g_mc$ cannot have infinite order, implying that they must have order 2. Take now, for $n_1\in I$, $g_{n_1}(g_{n_1}^{-1})^{-1} = g_{n_1}^2$: there must exist two different $g_{m_1}c,g_{m_2}c\in S$ such that $g_{m_1}g_{m_2}^{-1} = g_{n_1}^2$. Then $|J|\geq 2$. Now, by the same above argument of the maximal length, there exists $g_{m'}\in N$ such that $g_{m'}^{\pm 1}c\in S$. This implies that $g_{m'}\in S$ and finally, by isotropy, also
\begin{align*}
S =\lbrace \mathcal{O}_L(g_{m'}^{\pm 1}), \mathcal{O}_L(g_{m'}^{\pm 1})c \rbrace .
\end{align*}
Therefore one has:
\begin{align*}
e = (g_{m'}^{\pm 1}c)^2 = g_{m'}^{\pm 1}\varphi (g_{m'}^{\pm 1})c^{2}=(g_{m'}\varphi (g_{m'}))^{\pm 1}c^{2},
\end{align*}
implying $c^2 = e$ and $\varphi (g_{m''}) = g_{m''}^{-1}$ for all $g_{m''}\in \mathcal{O}_L(g_{m'})$. By Eq.~\eqref{ren-gen}, $\lbrace \mathcal{O}_L(g_{m'}^{\pm 1}), \mathcal{O}_L(\varphi (g_{m'})^{\pm 1}) \rbrace$ is a set of generators for $N$, and then $G$ must be then isomorphic to $\mathbb{Z}^d \rtimes_{\varphi} \mathbb{Z}_2$, with $\varphi(n)=-n$. By isotropy, denote now by $z_{\pm}$ and $z_c$ the transition scalars associated to the elements in $\mathcal{O}_L(g_{m'}^{\pm 1})$ and in $\lbrace\mathcal{O}_L(g_{m'})c,\mathcal{O}_L(g_{m'}^{-1})c\rbrace$, respectively. Computing the transition matrices of the coarse-grained QW on $\Gamma (N,S_N)$ using Eq.~\eqref{eq:coarsed_trans_matr}, these are equal to
\begin{align}\label{eq:coarsegrainedmatrices}
	A_{\pm h} = \begin{pmatrix}
		z_\pm  &  z_c  \\
		z_c   &  z_\mp
	\end{pmatrix}\quad \forall \pm h\in S_N.
\end{align}
For $d\geq 2$, it is easy to see that $|\mathcal{O}_L(g_{m'})|\geq 2$. Accordingly, in this case, in Eq.~\eqref{unitarity} one has at least one term with $h,h'\in S_{N}$ and $h\neq \pm h'$. Let us take the path $h-h'$ with maximal length, which is unique: this fact,
using Eq.~\eqref{unitarity} with Eq.~\eqref{eq:coarsegrainedmatrices}, implies the condition
\begin{align*}
|z_+|^2 + |z_c|^2 + |z_-|^2 + |z_c|^2 = 0,
\end{align*}
which is clearly impossible to satisfy. This proves the impossibility for all dimensions $d\geq 2$. In the case $d=1$, namely $N\cong \mathbb{Z}$, the only possible extension is isomorphic to $D_{\infty} = \mathbb{Z}\rtimes_{\varphi} \mathbb{Z}_2$. The most general family of scalar QWs on $D_\infty$ has been studied in Ref.~\cite{BDEPT16}. Imposing isotropy to these QWs, one straightforwardly finds that the resulting QWs contain the Weyl and Dirac QWs in one dimension.
\end{proof}
We conclude this investigation on the isotropic scalar QWs with the following theorem, which can be of aid in the attempt of generalizing the construction performed above.
\begin{theorem}[\citep{DK17}]
Let $G$ be a finitely generated group. Then every finite index subgroup $H$ in $G$ contains a subgroup $N$ which is finite
index and characteristic in $G$.
\end{theorem}
\begin{remark}
	Accordingly, a finite-by-$\mathbb{Z}^d$ group has, in general, the following structure:
\begin{align*}
G = \lbrace N,Nc_2,\ldots , Nc_i \rbrace,
\end{align*}
with the subgroup $N\cong \mathbb{Z}^d$ characteristic and of finite index $i$ in $G$. This, in particular, means that classifying finite-by-$\mathbb{Z}^d$ groups up to pseudo-congruence (see Proposition~\ref{prop:construction}) is equivalent to classifying virtually Abelian groups up to general isomorphisms. Moreover, Proposition~\ref{prop:autom} can be easily generalized to groups $G$ finite-by-$\mathbb{Z}^d$: in this case, $\mathrm{Aut}(G) \cong L_Q \rtimes \mathbb{GL}(d,\mathbb{Z})$, where $L_Q$ is the group acting on the coset representative as a pseudo-congruence. This fact, albeit not leading to a straightforward generalization of Proposition~\ref{prop:no-go}, can help to construct the Euclidean isotropic scalar QWs with index greater than 2.
\end{remark}

\subsection{Spinorial walks from scalar QWs}
In Subsec.~\ref{subsec:isotropic_scalar_QWs}, we imposed isotropy to those scalar QWs which reconstruct the QWs with a two-dimensional coin on lattices, finding that they exist in dimension $d=1$ only. In the present subsection we address the problem to provide necessary and suffient conditions in order to reconstruct a class of known spinorial QWs. In particular, we aim to derive, without imposing isotropy, QWs with a two-dimensional coin on the simple square lattice and on the BCC lattice. In fact, as already well understood in the Abelian case, unitarity alone in general does not allow to perform a tight \emph{a priori} selection on the admissible Cayley graphs. For this reason, in the following we restrict to the case of coarse-grained presentations exhibiting generators all with the same length in $\mathbb{R}^d$ (a necessary, but still not sufficient, condition for isotropy, as already pointed out). However, the one-dimensional case has been broadly studied in Ref.~\cite{BDEPT16}. Accordingly, as already mentioned, the cases which will be studied in generality shall be the simple square lattice and the BCC lattice.
\begin{lemma}\label{lem:semidirect}
Suppose that, for $d\in \mathbb{N}^+$ and $N\cong \mathbb{Z}^d$, a QW on $\Gamma (N,S_+)$, with a two-dimensional coin and having generators which can be embedded in $\mathbb{R}^d$ having all the same length, is the coarse-graining of a scalar QW on $\Gamma (G,S_+')$, with $G$ a $\mathbb{Z}_2$-by-$N$ non-Abelian group. Then $G\cong \mathbb{Z}^d\rtimes \mathbb{Z}_2$. 
\end{lemma}
\begin{proof}
Let us fix a set of generators $g_i\in S_N$ for $N\cong \mathbb{Z}^d$, with $\lVert \v{g}_{i_1}\rVert =\lVert \v{g}_{i_2}\rVert$ for all $i_1,i_2$. Denoting again $G=\lbrace N,Nc \rbrace$ and using Eq.~\eqref{ren-gen}, one can easily see that the following statements hold:
\begin{align}
&h_i \in S_G\cap N\ \Rightarrow\ h_i ,\varphi(h_i)\in S_N, \label{statement1} \\
&h_i= g_ic \in S_G\cap Nc \Rightarrow\ g_i,\varphi(g_i)+c^2 \in S_N. \label{statement2}
\end{align}
We can then follow the same argument in the first part of the proof of Proposition~\ref{prop:no-go}, since the same hypotheses hold. Therefore, we have that $\exists g_{i'}\in N \colon g_{i'}^{\pm 1}c,g_{i'} \in S_G$, and then also
\begin{align*}
(g_{i'}^{\pm 1}c)^{-1} = \varphi (g_{i'}^{\pm 1})c^{-2}c \in S_G.
\end{align*}
Then from Eq.~\eqref{statement2}, $\varphi (g_{i'}^{\pm 1})c^{-2}\in S_N$. Using also Eq.~\eqref{statement1}, we obtain $\lVert \pmb{\varphi} (\v{g_{i'}}) - \v{c^{2}} \rVert = \lVert -\pmb{\varphi} (\v{g_{i'}}) - \v{c^{2}} \rVert = \lVert \pmb{\varphi} (\v{g_{i'}}) \rVert$, implying $\v{c^2} = 0$ (i.e.~$c^2 = e$). Then $G$ is a semidirect product.
\end{proof}
We will consider two lattices for $\Gamma (N,S_+)$: the simple square and the BCC. The simple square lattice is generated by
\begin{align}\label{eq:sq_gen}
\v{h_1} = \begin{pmatrix}1\\0\end{pmatrix},\ \v{h_2} = \begin{pmatrix}0\\1\end{pmatrix},
\end{align}
while the BCC lattice is generated by
\begin{align}\label{eq:BCC_gen}
\v{h_1} = \begin{pmatrix}1\\1\\1\end{pmatrix},\ \v{h_2} = \begin{pmatrix}-1\\-1\\1\end{pmatrix},\ \v{h_3}= \begin{pmatrix}-1\\1\\-1\end{pmatrix},\ \v{h_4}= \begin{pmatrix}1\\-1\\-1\end{pmatrix}.
\end{align}
Following the same arguments of the proofs of Proposition~\ref{prop:no-go} and Lemma~\ref{lem:semidirect} (in particular, using Eqs.~\eqref{statement1} and~\eqref{statement2}), by direct inspection it is easy to see that the only set of generators for $G$ satisfying the quadrangularity condition (see Proposition~\ref{quadrang}) is of the form: \begin{align}\label{eq:admissible_presentation}
S' = \{ g, gc \}_{g\in S},
\end{align}
$S$ being the set of generators corresponding to the simple square or the BCC lattices.

Our strategy is the following. On  one hand, in Appendix~\ref{appendix} we derive the most general QWs with a two-dimensional coin on the simple square and BCC lattices. These are shown in Eqs.~\eqref{transmatr1} and~\eqref{transmatr2}. On the other hand, in Subsecs.~\ref{subsec:2D} and~\ref{subsec:3D} we construct the admissible groups (according to Lemma~\ref{lem:semidirect}) supporting the scalar QWs whose coarse-graining is a spinorial one on the two chosen lattices. Clearly, their coarse-graining is contained in the spinorial QWs derived there.
We will derive the form of the coarse-grained QWs and impose that form to the spinorial QWs~\eqref{transmatr1} and~\eqref{transmatr2} modulo unitary equivalence. This will allow us to find the most general family of scalar QWs on the extensions of $\mathbb{Z}_2 $ by $\mathbb{Z}^d$ ($d=2,3$), whose coarse-grainings are QWs with a two-dimensional coin on the simple square and BCC lattices.


\section{Relativistic wave equations from coinless QWs}\label{sec:relativistic_equations}

\subsection{One-dimensional case}
For completeness, we here briefly treat the one-dimensional case, namely $N\cong \mathbb{Z}$. The only nontrivial finite subgroup of $\mathrm{Aut}(N)\cong \mathbb{GL}(1,\mathbb{Z})$ is isomorphic to $\mathbb{Z}_2$ itself. Accordingly, the only non-Abelian index-2 extension by $\mathbb{Z}$ is isomorphic to the infinite dihedral group $D_\infty$. The most general family of admissible scalar QWs on $D_\infty$ has been studied in Ref.~\cite{BDEPT16}. This contains the one-dimensional Weyl and Dirac QWs.

We here outline a generalization of these results for higher indices.
\begin{theorem}[Lagrange~\cite{C12}]\label{lagrange}
Let $G$ be a finite group and $G'\leq G$. Then $|G'|$ divides $|G|$.\end{theorem}
This means that if we choose $|Q|$ to be odd, then the order $r_q$ of each element of $q\in Q$ will be odd. If we look for a nontrivial homomorphism $\varphi \colon Q \rightarrow \lbrace 1,-1 \rbrace$, with $|Q|$ odd, then for some $q$ we will have
\[
1 = \varphi_q^{r_q} = (-1)^{r_q} = -1,
\]
which is impossible. As a consequence, all the extensions of some group $Q$ of odd order by $\mathbb{Z}$ must satisfy
\begin{equation}\label{identity}
\varphi_q =1,\quad \forall q\in Q.
\end{equation}
\begin{example}[$\mathbb{Z}_3$-by-$\mathbb{Z}$ groups]
Let us pose $Q = \{ \tilde{e}, \tilde{q}, \tilde{q}^2 \}  \cong \mathbb{Z}_3$. Chosen a 2-cocycle $f$, one has
\begin{align*}
f(\tilde{q},\tilde{q})c_{\tilde{q}^2} &= c_{\tilde{q}}^2 =  (c_{\tilde{q}})c_{\tilde{q}}^2(c_{\tilde{q}})^{-1} = \\ &= \varphi_{\tilde{q}}\left( f(\tilde{q},\tilde{q}) \right) c_{\tilde{q}}c_{\tilde{q}^2}c_{\tilde{q}}^{-1} = f(\tilde{q},\tilde{q}) c_{\tilde{q}}c_{\tilde{q}^2}c_{\tilde{q}}^{-1},
\end{align*}
where we used Eq.~\eqref{identity}. This implies
\begin{align*}
c_{\tilde{q}^2}c_{\tilde{q}} = c_{\tilde{q}}c_{\tilde{q}^2}.
\end{align*}
However, by Eq.~\eqref{identity}, one also has 
\[
c_{\tilde{q}}n=nc_{\tilde{q}},\quad \forall n \in N,\forall \tilde{q}\in Q,
\]
whence it follows that an extension $G$ of $\mathbb{Z}_3$ by $\mathbb{Z}$ must be Abelian. Namely, $G$ is isomorphic to either $\mathbb{Z}$ or $\mathbb{Z}\times \mathbb{Z}_3$ (by the Fundamental theorem of finitely generated Abelian groups, see Theorem~\ref{fundab}).
\end{example}

\subsection{Two-dimensional case}\label{subsec:2D}
In Ref.~\cite{N72} one can find a complete classification of the finite subgroups of $\mathbb{GL}(2,\mathbb{Z})$ up to conjugation. The subgroups of order 2 are three, clearly isomorphic to $\mathbb{Z}_2 = \lbrace \tilde{e},\tilde{q} \rbrace$. These are, respectively, generated by:
\[
-I_2,\ 
\sigma_x =
\begin{pmatrix}
  0  &  1  \\
  1  &  0
\end{pmatrix},\ 
\sigma_z =
\begin{pmatrix}
  1  &  0  \\
  0  &  -1
\end{pmatrix}.
\]
The simple square lattice is generated by the canonical basis for $\mathbb{R}^2$, i.e.~the set of vectors~\eqref{eq:sq_gen}. Regardless of the homomorphism $\varphi$  chosen, one has that $c_{\tilde{q}}^2\in N\cong \mathbb{Z}^2$ will be invariant under $\varphi_{\tilde{q}}$. For each automorphism, then, one has to fix the only nontrivial 2-cocycle value $f(\tilde{q},\tilde{q}) = c_{\tilde{q}}^2$ among the invariant vectors of the chosen automorphism. In the following, we provide an application of the group-extension technique developed in Sec.~\ref{sec:extension_problem}. We construct all the index-2 non-Abelian extensions of $N\cong \mathbb{Z}^2$ below.
\begin{itemize}
\item \textbf{Case $\varphi_{\tilde{q}} = -I_2$.} The only invariant element is the zero vector. Then the only possibility is $J_1 \cong \mathbb{Z}^2 \rtimes_{\mathrm{-I_2}} \mathbb{Z}_2$.
\item \textbf{Case $\varphi_{\tilde{q}} = \sigma_x$.} The invariant space in $\mathbb{Z}^2$ is $r_x (\v{h_1}+\v{h_2})$, with $r_x\in \mathbb{Z}$. Let us choose $r_x=0$: using Eq.~\eqref{cosetchange2}, we have
\begin{align*}
(I_2 + \sigma_x)(n_{\tilde{q},0}\v{h_1} + n_{\tilde{q},1}\v{h_2})+r'_x (\v{h_1}+\v{h_2}) &= \\ 
 (n_{\tilde{q},0} + n_{\tilde{q},1} + r'_x)(\v{h_1}+\v{h_2}) = 0,&
\end{align*}
namely $\forall r'_x\ \exists n_{\tilde{q},0},n_{\tilde{q},1}\in \mathbb{Z}$ such that the extension with ${c'}_{\tilde{q}}^2 = r'_x (\v{h_1}+\v{h_2})$ differs from that with $c_{\tilde{q}}^2 = 0$ for a change of coset representative. Then the only possibility is again $J_2 \cong \mathbb{Z}^2 \rtimes_{\sigma_x} \mathbb{Z}_2$.
\item \textbf{Case $\varphi_{\tilde{q}} = \sigma_z$.} The invariant space is $r_z \v{h_1}$, with $r_z\in \mathbb{Z}$. Choosing $r_z=0$, again from Eq.~\eqref{cosetchange2} we have
\begin{align*}
(I_2 + \sigma_z)(n_{\tilde{q},0}\v{h_1} + n_{\tilde{q},1}\v{h_2})+r'_z \v{h_1} = \\ = (2n_{\tilde{q},0} + r'_z)\v{h_1} = 0,
\end{align*}
namely, $\forall r'_z$ even, $\exists n_{\tilde{q},0},n_{\tilde{q},1}\in \mathbb{Z}$ such that the extension with ${c'}_{\tilde{q}}^2 = r'_z \v{h_1}$ differs from that with $c_{\tilde{q}}^2 = 0$ for a change of coset representative. Then $J_3 \cong \mathbb{Z}^2 \rtimes_{\sigma_z} \mathbb{Z}_2$. Choosing $r_z=1$, one has
\[
(2n_{\tilde{q},0} + r'_z)\v{h_1} = \v{h_1}
\]
and $\forall r'_z$ odd $\exists n_{\tilde{q},0},n_{\tilde{q},1}\in \mathbb{Z}$ such that the extension with ${c'}_{\tilde{q}}^2 = r'_z \v{h_1}$ differs from that with $c_{\tilde{q}}^2 =\v{h}_1$ for a change of coset representative. In this last case, the extension $J_4$ is not a semidirect product.
\end{itemize}
We list some possible presentations of all the non-Abelian extensions of $\mathbb{Z}_2$ by $\mathbb{Z}^2$:
\begin{align*}
J_1 &= \< h_1,h_2,c \left.|\right. h_1h_2h_1^{-1}h_2^{-1}, c^2, ch_1c^{-1}h_1, ch_2c^{-1}h_2 \>, \\
J_2 &= \< h_1,h_2,c \left.|\right. h_1h_2h_1^{-1}h_2^{-1}, c^2, ch_1c^{-1}h_2^{-1}, ch_2c^{-1}h_1^{-1} \>, \\
J_3 &= \< h_1,h_2,c \left.|\right. h_1h_2h_1^{-1}h_2^{-1}, c^2, ch_1c^{-1}h_1^{-1}, ch_2c^{-1}h_2 \>, \\
J_4 &= \< h_2,c \left.|\right. c^2h_2c^{-2}h_2^{-1}, ch_2c^{-1}h_2 \>.
\end{align*}
We now construct the admitted scalar QWs on the derived extensions. Lemma~\ref{lem:semidirect} excludes $J_4$, since it is not a semidirect product. The admissible presentations are of the form~\eqref{eq:admissible_presentation}, $S =\lbrace h_1,-h_1,h_2,-h_2\rbrace$ being the set of generator associated to the simple square lattice. We write the coarse-grained transition matrices, derived using Eq.~\eqref{eq:coarsed_trans_matr}, for the three cases:
\begin{enumerate}
	\item Case of $J_1$, for $x\in \{h_1,h_2\}$:
	\begin{gather*}\label{eq:scalar_transmatr1}
	A_{+x}=
	\begin{pmatrix}
	z_{x}  &  z_{xc}  \\
	z_{x^{-1}c}  &  z_{x^{-1}}
	\end{pmatrix},\ 
	A_{-x}=
		\begin{pmatrix}
		z_{x^{-1}} &  z_{x^{-1}c}  \\
		z_{xc}  &  z_{x}
		\end{pmatrix}.
		\end{gather*}
	\item Case of $J_2$, for $x,y\in \{h_1,h_2\} : x\neq y$:
		\begin{gather*}\label{eq:scalar_transmatr2}
	A_{+x}=
	\begin{pmatrix}
	z_{x}  &  z_{xc}  \\
	z_{yc}  &  z_{y}
	\end{pmatrix},\ 
	A_{-x}=
	\begin{pmatrix}
	z_{x^{-1}} &  z_{x^{-1}c}  \\
	z_{y^{-1}c}  &  z_{y^{-1}}
	\end{pmatrix}.
	\end{gather*}
	\item Case of $J_3$:
	\begin{gather*}\label{eq:scalar_transmatr3}
	A_{\pm h_1}=
	\begin{pmatrix}
	z_{h_1^{\pm 1}}  &  z_{h_1^{\pm 1}c}  \\
	z_{h_1^{\pm 1}c}  &  z_{h_1^{\pm 1}}
	\end{pmatrix},\ 
	A_{\pm h_2}=
	\begin{pmatrix}
	z_{h_2^{\pm 1}} &  z_{h_2^{\pm 1}c}  \\
	z_{h_2^{\mp 1}c}  &  z_{h_2^{\mp 1}}
	\end{pmatrix}.
	\end{gather*}
\end{enumerate}
One can check the  conditions under which the above coarse-grained matrices are unitarily equivalent to those of~\eqref{transmatr}.
\begin{enumerate}
	\item Every QW obtained as the two-dimensional Weyl QW~\cite{PhysRevA.90.062106}, multiplied on the left by an arbitrary unitary commuting with $\sigma_y$ can be viewed as a coarse-graining of  a scalar QW on the Cayley graph corresponding to presentation  $J_1$. 
	\item Every isotropic 2D QW with a two-dimensional coin~\cite{d2017isotropic} can be viewed as a coarse-graining of a scalar QW on the Cayley graph corresponding to presentation  $J_2$.
	\item No QW with a two-dimensional coin the simple square lattice can be viewed as  a coarse-graining of a scalar  QW on the Cayley graph corresponding to presentation $J_3$.
\end{enumerate}
Both families of coinless QWs on $J_1$ and on $J_2$ stricly contain the Weyl QW in two-space dimensions~\cite{PhysRevA.90.062106,d2017isotropic}. The presentations corresponding to the QWs on the Cayley graphs of $J_2$ and its coarse-graining, along with, respectively, the associated transition scalar and matrices, are reported in Fig.~\ref{fig:Lattices}.
\begin{figure*}[t]
	\subcaptionbox{\label{fig:001}}{%
		\includegraphics[width=.455\linewidth]{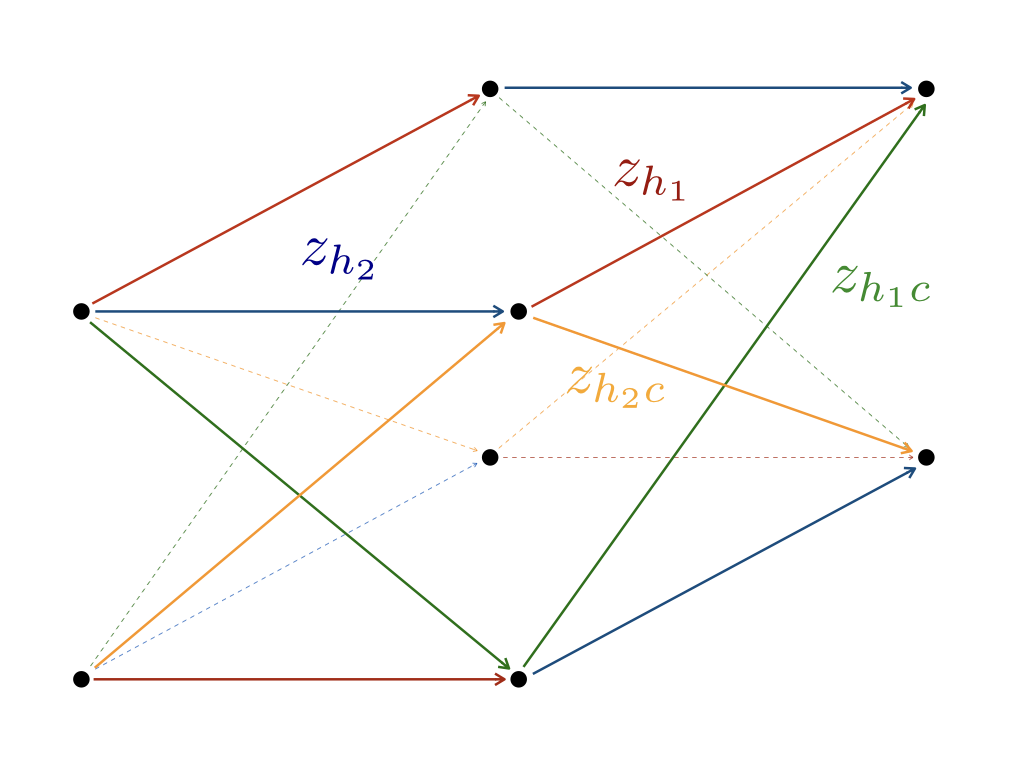}}\qquad
	\subcaptionbox{\label{fig:002}}{%
		\includegraphics[width=.455\linewidth]{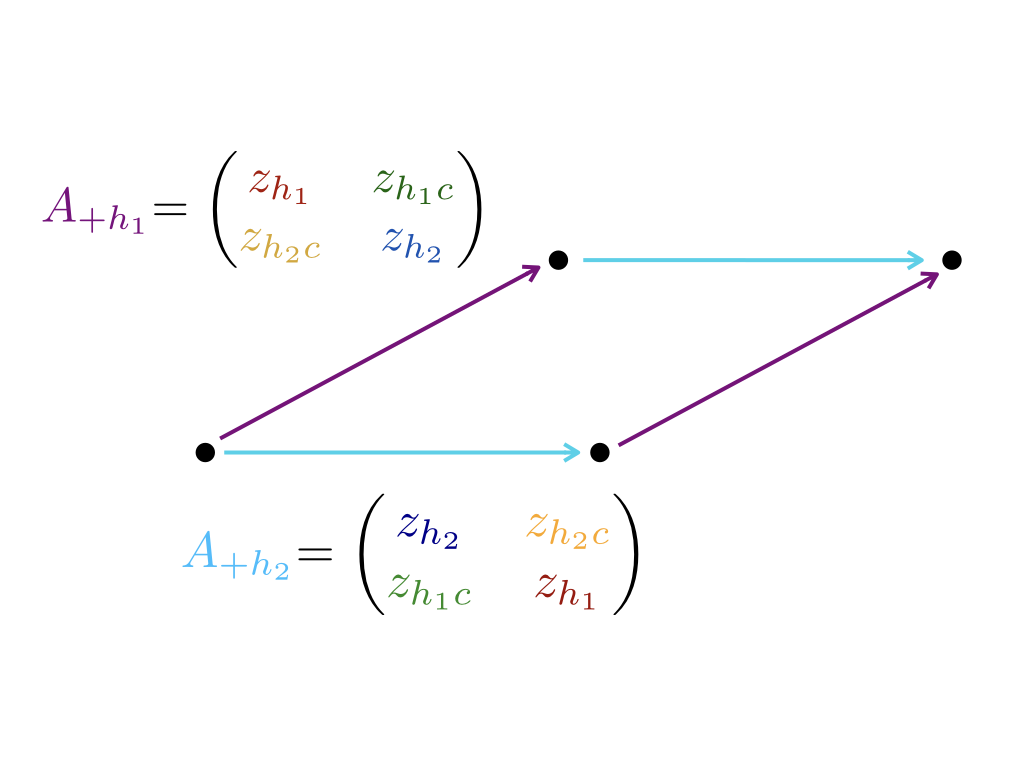}}
	\caption{We report here the primitive cell of the coinless QW on the extension  $J_2$ (Figure \ref{fig:001}), and that of the same walk on $J_2$ represented as a spinorial coarse-grained QW on $\mathbb{Z}^2$ (Figure \ref{fig:002}).}
	\label{fig:Lattices}
\end{figure*}

\subsection{Three-dimensional case}\label{subsec:3D}
In Ref.~\cite{T71} one can find a complete classification of the finite subgroups of $\mathbb{GL}(3,\mathbb{Z})$ up to conjugation. One can exploit the same methods provided in the two-dimensional case (Subsec.~\ref{subsec:2D}) in order to construct all the extension of $\mathbb{Z}_2$ by $\mathbb{Z}^3$. 

The subgroups of order 2 are five, clearly isomorphic to $\mathbb{Z}_2$, and are respectively generated by:
\begin{align}\label{eq:3D_automorphisms}
-I_3,\ 
\Sigma_\pm \coloneqq
\pm \begin{pmatrix}
  1  &  0  &  0  \\
  0  & -1  &  0   \\
  0  &  0  &  -1
\end{pmatrix},\ 
\Lambda_\pm \coloneqq
\pm \begin{pmatrix}
  -1 &  0  &  0  \\
  0  &  0  &  1   \\
  0  &  1  &  0
\end{pmatrix}.
\end{align}
Lemma~\ref{lem:semidirect} excludes the extensions which are not semidirect products. Accordingly, all the admissible non-Abelian extensions of $\mathbb{Z}_2$ by $\mathbb{Z}^3$ are derived as follows:
\begin{enumerate}
\item $K_1 \cong \mathbb{Z}^3 \rtimes_{-I_3} \mathbb{Z}_2$,
\item $K_2 \cong \mathbb{Z}^3 \rtimes_{\Sigma_+} \mathbb{Z}_2$,
\item $K_3 \cong \mathbb{Z}^3 \rtimes_{\Sigma_-} \mathbb{Z}_2$,
\item $K_4 \cong \mathbb{Z}^3 \rtimes_{\Lambda_+} \mathbb{Z}_2$,
\item $K_5 \cong \mathbb{Z}^3 \rtimes_{\Lambda_-} \mathbb{Z}_2$.
\end{enumerate}
The BCC lattice is generated by the set of vectors~\eqref{eq:BCC_gen}. The presentation of the groups $K_i$ admitting a scalar QW is given by~\eqref{eq:admissible_presentation}, with $S$ collecting the generators~\eqref{eq:BCC_gen} and the element $c$ realizing the automorphisms~\eqref{eq:3D_automorphisms} by conjugation. Adopting the same arguments used in Subsec.~\ref{subsec:2D}, one finds that $K_4$ and $K_5$ do not support a scalar QW whose coarse-graining is a QW with a two-dimensional coin on the BCC lattice. On the other hand, $K_1$, $K_2$ and $K_3$ admit scalar QWs, which can be found directly checking unitary equivalence. In particular, the transition matrices in the case of $K_2$ are:
\begin{gather*}\label{eq:scalar_transmatr1}
	A_{\pm h_1}=
	\begin{pmatrix}
	z_{h_1^{\pm 1}}  &  z_{h_1^{\pm 1}c}  \\
	z_{h_4^{\pm 1}c}  &  z_{h_4^{\pm 1}}
	\end{pmatrix},\ 
	A_{\pm h_2}=
	\begin{pmatrix}
	z_{h_2^{\pm 1}}  &  z_{h_2^{\pm 1}c}  \\
	z_{h_3^{\pm 1}c}  &  z_{h_3^{\pm 1}}
	\end{pmatrix}, \\
			A_{\pm h_3}=
	\begin{pmatrix}
	z_{h_3^{\pm 1}}  &  z_{h_3^{\pm 1}c}  \\
	z_{h_2^{\pm 1}c}  &  z_{h_2^{\pm 1}}
	\end{pmatrix},\ 
	A_{\pm h_4}=
	\begin{pmatrix}
	z_{h_4^{\pm 1}}  &  z_{h_4^{\pm 1}c}  \\
	z_{h_1^{\pm 1}c}  &  z_{h_1^{\pm 1}}
	\end{pmatrix}.
		\end{gather*}
Accordingly, there exists a unitary $W$ such that $WA_{\pm h_1}W^\dagger = A_{\pm h_4}$ and $WA_{\pm h_2}W^\dagger = A_{\pm h_3}$, and the family of scalar QWs on $K_2$ strictly contains the Weyl QW in three space-dimensions~\cite{PhysRevA.90.062106,d2017isotropic}.

\subsection{The 2D and 3D Dirac QWs from coinless walks}\label{subsec:Dirac}
For one space-dimension, we already saw in Proposition~\ref{prop:no-go} that the 1D Weyl and Dirac QWs are the coarse-graining of isotropic scalar walks on $D_\infty$ (see the analysis carried out in Ref.~\cite{BDEPT16}).

We conclude the present work constructing two extensions of $\mathbb{Z}_2\times \mathbb{Z}_2$ by $\mathbb{Z}^d$ for $d=2,3$. This is needed to implement the coinless Dirac QWs~\cite{PhysRevA.90.062106}. Clearly, the coin dimension of the coarse-grained QWs must be of dimension 4, and this is the reason why we take an index-4 extension by $\mathbb{Z}^d$.

Let $G_{d}$ be an extension of $Q\cong D_2 = \mathbb{Z}_2\times \mathbb{Z}_2$ by $N^d\cong \mathbb{Z}^d$, with
\begin{align*}
G_d = \lbrace N,Nc_1,Nc_2,Nc_1c_2 \rbrace
\end{align*}
and
\begin{align*}
c_1^2 = c_2^2 = e,\quad c_1c_2=c_2c_1.
\end{align*}
Accordingly, $G_d \cong \mathbb{Z}^d\rtimes_\varphi D_2$. The associated automorphisms realizing the semidirect product are denoted by $\varphi_1$, $\varphi_2$ and $\varphi_{12} = \varphi_1 \circ \varphi_2 = \varphi_2 \circ \varphi_1$.

We choose the following presentation for $G_d$:
\begin{align*}
S'_d = \lbrace h, hc_1, c_2 \left.|\right. h\in S_d  \rbrace,
\end{align*}
where $S_d$ is the set of generators corresponding to the simple square lattice for $d=2$, and to the BCC lattice for $d=3$ (see Eqs.~\eqref{eq:sq_gen} and~\eqref{eq:BCC_gen}). Let $\lbrace \varphi_1,\varphi_2 \rbrace$ be represented in dimension $d=2$ by
\begin{align*}
\lbrace \sigma_z,\ -\sigma_z \rbrace,\quad
\sigma_z = \begin{pmatrix}
  1 &  0  \\
  0  &  -1
\end{pmatrix},
\end{align*}
while in dimension $d=3$ by
\begin{align*}
\lbrace \Sigma_+,\ \Sigma_+' \rbrace,\quad
\Sigma_+ = \begin{pmatrix}
  1 &    0  \\
  0  &  -I_2
\end{pmatrix},\quad
\Sigma_+' = \begin{pmatrix}
  -I_2 &    0  \\
  0    &  1
\end{pmatrix}.
\end{align*}
It is immediate to check that the associated Cayley graphs satisfy the quadrangularity condition (Proposition~\ref{quadrang}). Accordingly, via Eq.~\eqref{eq:coarsed_trans_matr}, the coarse-grained transition matrices are given by:
\begin{align}\label{eq:dirac_scalar}
\begin{split}
A^{D}_h &= 
\begin{pmatrix}
  nA_h         &   0  \\
  0  &   nB_h
\end{pmatrix},\quad A^{D}_e = 
\begin{pmatrix}
  0        &   z_{c_2}I_d  \\
  z_{c_2}I_d  &   0
\end{pmatrix}, \\ A_h &= 
\frac{1}{n}\begin{pmatrix}
  z_{h}              & z_{hc_1}   \\
  z_{\varphi_1(h)c_1}   & z_{\varphi_1(h)} 
\end{pmatrix},\\ B_h &=
\frac{1}{n}\begin{pmatrix}
  z_{\varphi_2(h)}  & z_{\varphi_2(h)c_1}   \\
  z_{\varphi_{12}(h)c_1}  & z_{\varphi_{12}(h)} 
\end{pmatrix},
\end{split}
\end{align}
where $n\in (0,1]$ and $A_e^D$ is associated to the identity element (i.e.~the mass term). Since the form of the matrices~\eqref{eq:dirac_scalar} is the same of those of the Dirac QWs~\cite{PhysRevA.90.062106}, we can  choose the $A_h$ to be the transition matrices of Subsecs.~\ref{subsec:2D} and~\ref{subsec:3D} for the two- and three-dimensional Weyl scalar QWs, with $B_h = A_h^\dagger$ and $z_{c_2} = im$, with $m\geq 0$ and $n^2+m^2=1$. We notice that one can perform a change of basis only on the diagonal blocks, attaining the desired form. These choices guarantee unitarity and select the $d$-dimensional Dirac scalar QWs for $d=2,3$.

\section{Discussion and Conclusions}\label{sec:conclusions}
In this work we addressed the problem of providing a constructive procedure to obtain all the possible QWs on Cayley graphs quasi-isometric to Euclidean spaces---i.e.~of virtually Abelian groups. The latter are exhausted by groups that are extensions of finite groups by free Abelian groups. The problem has then been tackled starting from a thorough characterization of all extensions of finite groups by arbitrary groups, providing (partial) criteria for checking isomorphism of a priori different extensions. The piece of theory thus constructed is then applied in the special case of interest, i.e.~where the extension is by a free Abelian group.

We then constructed all the extensions of $\mathbb Z_2$ by $\mathbb Z$ and $\mathbb Z^2$. As to the case of $\mathbb Z_2$ by $\mathbb Z^3$, the analysis has been restricted to semidirect products, even though there exist also different extensions. However, extensions of the latter kind would lead to embeddings in $\mathbb R^3$ where the nearest neighbours of a given node need to have different distances from the node. We then neglected extensions that are not semidirect products, since the subsequent analysis focused on the simple square and the BCC, where all the nearest neighbours of a given node have the same distance from the node.

We then proved that there are no isotropic scalar QWs on (non-Abelian) extensions of $\mathbb Z_2$ by free Abelian groups of dimension larger than 1. Our main results consist in providing criteria to find the families of QWs with $s=2$ on the simple square and BCC lattices that can be obtained by coarse-graining a QW on $\mathbb Z_2$-by-$\mathbb Z^d$ with $d=2,3$. Among those, interestingly, there are the Weyl QWs in 2 and 3 dimensions. Finally, in a very special case of the extension $(\mathbb Z_2\times\mathbb Z_2)$-by-$\mathbb Z^d$ with $d=2,3$, we show that the Dirac QWs in $d=2,3$ can be obtained by coarse-graining.

This line of research has a relevance within our approach to QFTs, where that all the non-trivial structures can be derived from very simple evolution algorithms. In particular, internal degrees of freedom like spin or chirality can be viewed as obtained by coarse-graining of QWs with trivial coin, i.e.~originally describing the dynamics of particles with trivial internal structure. A big problem that remains open in this respect is the origin of symmetries like isotropy in this scenario. One of the possible developments might consist in studying how isotropy of coarse-grained QWs can follow from other requirements for the coinless underlying QWs, e.g.~those considered in Ref.~\cite{PhysRevA.75.062332}. Along these lines, one can imagine a relaxation of isotropy where the symmetry involves only a family of subgraphs.

We characterized Euclidean QWs in $d$ dimensions, proving that these are exhausted by the QWs on $\mathbb Z^d$ via a unitary coarse-graining procedure. One possible future line of research is the investigation of the properties of the Euclidean QWs which are not coarse-graining of any (scalar) QW, e.g.~a subfamily of those derived in the Appendix of the present work. In particular, the dynamics and space-time symmetries of these can be in principle different from the ones exhibited by the Weyl and Dirac QWs. Moreover, one could ask whether or not there exist Euclidean scalar QWs on Cayley graphs of (non-Abelian) groups with no cyclic element (i.e.~without \emph{torsion}) under the hypothesis of coarse-grained generators having all the same length. Finally, the next step is a characterization and study of QWs on Cayley graphs which are not embeddable in Euclidean spaces, i.e.~carrying a non-trivial curvature.

\section*{Acknowledgments}
The authors wish to thank Benson S. Farb, Cornelia Dru{\c{t}}u, Roberto Frigerio, Romain Tessera, and Francesco Genovese for useful discussions and precious insights. The section presenting original results on group extension benefited from the material of Ref.~\cite{morandi}. This publication was made possible through the support of a grant from the John Templeton Foundation, ID \# 60609 ``Quantum Causal Structures''. The opinions expressed in this publication are those of the authors and do not necessarily reflect the views of the John Templeton Foundation.

\appendix

\section{Derivation of the QWs}\label{appendix}
In Appendices~\ref{appendix:sql} and~\ref{appendix:bcc} we derive the most general families of QWs with a two-dimensional coin on the simple square and BCC lattices.

\subsection{Simple square lattice}\label{appendix:sql}
The presentation of $\mathbb{Z}^2$ corresponding to the simple square lattice is given by: $\< h_1,h_2 | h_1h_2h_1^{-1}h_2^{-1} \>$. The unitarity conditions~\eqref{unitarity} are given by:
\begin{equation}\label{eq:unit}
\begin{aligned}
A_{+i}A_{-i}^{\dagger} = A_{-i}^{\dagger}A_{+i} = 0, \\
A_{+i}A_{\mp j}^{\dagger} + A_{\pm j}A_{-i}^{\dagger} = 0, \\
A_{-i}^{\dagger}A_{\pm j} + A_{\mp j}^{\dagger}A_{+i} = 0 ,
\end{aligned}
\end{equation}
for $i,j\in \lbrace h_1,h_2 \rbrace$ and $i \neq j$. Let us introduce the \emph{polar decomposition} for the transition matrices: $\forall A \in \mathbb{GL}(n,\mathbb{C})$ there exists $V$ unitary and $P$ positive semidefinite such that $A=VP$. The first of Eqs.~\eqref{eq:unit} implies
\begin{align*}
A_{\pm i} = \alpha_{\pm i} V_i \ket{\pm i} \bra{\pm i},
\end{align*}
where $\{ \ket{-i}, \ket{+i} \}$ are orthonormal basis and $\alpha_{\pm i} >0$ $\forall i$ and we posed $V_{+i} = V_{-i} \eqqcolon V_i$: It is easy to verify it in view of the non-uniqueness of the polar decomposition when the matrix is not full-rank (we refer the reader to Ref.~\cite{DEPT16}).

The second and third of Eqs.~\eqref{eq:unit} imply
\begin{align*}
A_{+i}A_{\pm j}^{\dagger}A_{+i} = A_{-i}A_{\pm j}^{\dagger}A_{-i} = 0,
\end{align*}
reading
\begin{align*}
\braket{+i}{\mp j}\bra{\mp j}V_j^{\dagger}V_i\ket{+i} = \braket{-i}{\pm j}\bra{\pm j}V_j^{\dagger}V_i\ket{-i} = 0 .
\end{align*}
It's easy to see that at least one among $\bra{\mp j}V_j^{\dagger}V_i\ket{+i}$ and $\bra{\pm j}V_j^{\dagger}V_i\ket{-i}$ must be vanishing. The cases are just two:
\begin{enumerate}
\item \label{c1} $\bra{\pm j}V_j^{\dagger}V_i\ket{\pm i} = 0 \Rightarrow \ket{\pm i} = \ket{\pm j} \coloneqq \ket{0},\ket{1}$,
\item \label{c2} $\bra{\mp j}V_j^{\dagger}V_i\ket{\pm i} = 0 \Rightarrow \ket{\pm i} = \ket{\mp j} \coloneqq \ket{0},\ket{1}$
\end{enumerate}
(up to a phase factor which is not relevant in defining the transition matrices). Then $V_j^{\dagger}V_i$ is anti-diagonal in the basis $\lbrace \ket{0},\ket{1} \rbrace$, say $V_1 \eqqcolon W = V \begin{pmatrix}
0 & \mu \\ \nu & 0 \end{pmatrix}$, where $V \coloneqq V_2$. Using the second and third of Eqs.~\eqref{eq:unit} and Eqs.~\eqref{unitarity} for $g=e$ (i.e.~the normalization condition), one gets in both cases $\mu = -\nu^*$, $\alpha_{+1} = \alpha_{-1} \eqqcolon \alpha$ and $\alpha_{+2} = \alpha_{-2} = \sqrt{1-\alpha^2}$. Substituting and changing basis to set $\nu = 1$, one gets
\begin{align}\label{transmatr1}
\begin{split}
A_{+1} &= \alpha V \ket{1}\bra{0},\ A_{-1} = -\alpha V \ket{0}\bra{1}, \\ A_{+2} &= \sqrt{1-\alpha^2} V \ket{0}\bra{0},\ A_{-2} = \sqrt{1-\alpha^2} V \ket{1}\bra{1},
\end{split}
\end{align}
for the first case, while the second case is recovered just swapping $+2 \leftrightarrow -2$.

Now, one can check that the unitarity conditions
\begin{equation}
\begin{split}
A_{h}A_e^\dagger + A_eA_{-h}^\dagger = 0, \\
A_e^\dagger A_{h} + A_{-h}^\dagger A_e = 0
\end{split}
\end{equation}
cannot be satisfied for all $h\in S_+$.

Finally, can impose the condition~\eqref{eq:sum_unit}, finding $V = \begin{pmatrix}
\sqrt{1-\alpha^2} & \alpha \\
     -\alpha   &  \sqrt{1-\alpha^2}
\end{pmatrix}$. Plugging this into the transition matrices~\eqref{transmatr1} one obtains:
\begin{equation}\label{transmatr}
\begin{aligned}
A_{+1} = \begin{pmatrix}
\alpha^2 & 0 \\ \alpha\sqrt{1-\alpha^2} & 0
\end{pmatrix},\ A_{-1} = \begin{pmatrix}
0 & -\alpha\sqrt{1-\alpha^2} \\ 0 & \alpha^2\end{pmatrix}, \\ A_{+2} = \begin{pmatrix}
1-\alpha^2 & 0 \\ -\alpha\sqrt{1-\alpha^2} & 0
\end{pmatrix},\ A_{-2} = \begin{pmatrix}
0 & \alpha\sqrt{1-\alpha^2} \\ 0 & 1-\alpha^2\end{pmatrix},
\end{aligned}
\end{equation}
with $\alpha\in (0,1)$, and up to a left-multiplication by an arbitrary unitary commuting with the symmetry group giving the invariance condition~\eqref{eq:isotropic_invariance}.

\subsection{BCC lattice}\label{appendix:bcc}
The presentation of $\mathbb{Z}^3$ corresponding to the BCC lattice is given by: $\< h_1,h_2,h_3,h_4| h_ih_jh_i^{-1}h_j^{-1}, h_1h_2h_3h_4 \>$. There are three kinds of different paths of length 2 giving rise to the unitarity conditions:
\begin{align}
&\pm 2h_i, \label{cond1} \\
&\pm h_i \mp h_j, \label{cond2} \\
&\pm (h_i+h_j), \label{cond3}
\end{align}
for $h_i,h_j \in S_+$. Similarly to the two-dimensional case, from the first family of paths~\eqref{cond1} we obtain the following general expression for the transition matrices:
\begin{equation} \label{trans_matr}
A_{\pm i} = \alpha_{\pm i} V_i \ket{\pm i}\bra{\pm i},
\end{equation}
with $\{ \ket{-i},\ket{+i} \}$ an orthonormal basis and $\alpha_{\pm i} > 0$ $\forall i$. On the other hand, the condition associated to~\eqref{cond2} amounts to
\begin{equation}\label{eq:unit2}
\begin{aligned}
A_{+i}A_{+j}^{\dagger} + A_{-j}A_{-i}^{\dagger} = 0, \\
A_{+i}^{\dagger}A_{+j} + A_{-j}^{\dagger}A_{-i} = 0 .
\end{aligned}
\end{equation}
Exploiting the form~\eqref{trans_matr} and using Eqs.~\eqref{eq:unit2}, we get:
\begin{equation}\label{cond4}
A_{+i}A_{+j}^{\dagger}A_{+i} = 0\ \Rightarrow\ A_{+i}A_{+j}^{\dagger} = 0\ \vee\ A_{+j}^{\dagger}A_{+i} = 0 .
\end{equation}
Accordingly, one of the two following cases hold:
\begin{align}
\begin{split}
A_{\pm i} = \alpha_{\pm i} V_i \ket{\pm i}\bra{\pm i},\\
A_{\pm j} = \alpha_{\pm j} V_j \ket{\mp i}\bra{\mp i};\label{case1}
\end{split}
\\
\begin{split}
A_{\pm i}^{\dagger} = \alpha_{\pm i} V_i^{\dagger}V_i \ket{\pm i}\bra{\pm i}V_i^{\dagger},\\
A_{\pm j}^{\dagger} = \alpha_{\pm j} V_j^{\dagger} V_i\ket{\mp i}\bra{\mp i} V_i^{\dagger}. \label{case2}
\end{split}
\end{align}
In case~\eqref{case1}, $A_{+j}A_{+i}^{\dagger} = A_{+l}A_{+i}^{\dagger} = 0$ implies $A_{+j}A_{+l}^{\dagger} \neq 0$, while in case~\eqref{case2} $A_{+i}^{\dagger}A_{+j} = A_{+i}^{\dagger}A_{+l} = 0$ implies $A_{+j}^{\dagger}A_{+l} \neq 0$: since there are four elements in $S_+$, it is easy to see that, for a fixed $+i$, conditions in~\eqref{cond4} can be satisfied at most for two different values of $+j$. The possible couples are six, then either the first or the second condition must be satisfied for at least two couples with a fixed $+i$. Thus one has (modulo relabeling the $h_l$) three set of conditions:
\begin{equation}\label{cond5}
\begin{aligned}
A_{+1}A_{+2}^{\dagger} = A_{+1}A_{+3}^{\dagger} = A_{+2}A_{+4}^{\dagger}  = 0, \\
A_{+2}^{\dagger}A_{+3} = A_{+1}^{\dagger}A_{+4} = A_{+3}^{\dagger}A_{+4} = 0 ,
\end{aligned}
\end{equation}
or
\begin{equation}\label{cond6}
\begin{aligned}
A_{+1}A_{+2}^{\dagger} = A_{+1}A_{+3}^{\dagger} = A_{+2}A_{+4}^{\dagger} = A_{+3}A_{+4}^{\dagger}  = 0, \\
A_{+2}^{\dagger}A_{+3} = A_{+1}^{\dagger}A_{+4} = 0 ,
\end{aligned}
\end{equation}
or the previous one modulo the exchange of $A_{+i}$ and $A_{+i}^{\dagger}$, i.e.~equivalently modulo the PT transformation $A_{\v{k}} \mapsto A_{\v{k}}^{\dagger}$. It is then sufficient to solve the first two sets of equations.

Imposing the conditions which are common to Eqs.~\eqref{cond5} and~\eqref{cond6}, we obtain:
\begin{equation*}
\begin{aligned}
A_{+1} &= \alpha_{+1}V_1 M,& A_{-1} &= \alpha_{-1}V_1 (1-M), \\
A_{+2} &= \alpha_{+2}V_2 (1-M),& A_{-2} &= \alpha_{-2}V_2 M, \\
A_{+3} &= \alpha_{+3}V_3 (1-M),& A_{-3} &= \alpha_{-3}V_3 M, \\
A_{+4} &= \alpha_{+4}V_4 M,& A_{-4} &= \alpha_{-4}V_4 (1-M),
\end{aligned}
\end{equation*}
where $M \coloneqq \ket{0}\bra{0}, 1-M \coloneqq \ket{1}\bra{1}$ are arbitrary one-dimensional projectors, and $V_1^{\dagger}V_4,V_2^{\dagger}V_3$ have vanishing diagonal elements. This form of the solutions is equivalent to the set of constraints~\eqref{cond6}. Imposing~\eqref{eq:unit2} we find
\begin{equation*}
\begin{aligned}
\alpha_{+1}\alpha_{+4} V_1 M V_4^{\dagger} + \alpha_{-1}\alpha_{-4} V_4 (1-M) V_1^{\dagger} = 0, \\
\alpha_{+2}\alpha_{+3} V_2 (1-M) V_3^{\dagger} + \alpha_{-2}\alpha_{-3} V_3 M V_2^{\dagger} = 0,
\end{aligned}
\end{equation*}
implying that
\begin{equation}\label{sing_val1}
\begin{aligned}
\alpha_{+1} \alpha_{+4} = \alpha_{-1}\alpha_{-4}, \\
\alpha_{+2} \alpha_{+3} = \alpha_{-2}\alpha_{-3}, \\
V_1^{\dagger}V_4, V_2^{\dagger}V_3 \in \mathbb{SU}(2).
\end{aligned}
\end{equation}
Also, we have
\begin{equation}\label{cond7}
\begin{aligned}
\alpha_{+1}\alpha_{+2} MV_1^{\dagger}V_2 (1-M) + \alpha_{-1}\alpha_{-2} MV_2^{\dagger}V_1 (1-M)  = 0, \\
\alpha_{+1}\alpha_{+3} MV_1^{\dagger}V_3 (1-M) + \alpha_{-1}\alpha_{-3} MV_3^{\dagger}V_1 (1-M)  = 0, \\
\alpha_{+2}\alpha_{+4} (1-M) V_2^{\dagger}V_4 M + \alpha_{-2}\alpha_{-4} (1-M)V_4^{\dagger}V_2 M  = 0, \\
\alpha_{+3}\alpha_{+4} (1-M) V_3^{\dagger}V_4 M + \alpha_{-3}\alpha_{-4} (1-M)V_4^{\dagger}V_3 M  = 0,
\end{aligned}
\end{equation}
implying $(V_i^{\dagger}V_j)_{01} = - (V_i^{\dagger}V_j)_{10}^*$ for the pairs $(i,j) = (1,2),(1,3),(2,4),(3,4)$. Then we can pose
\begin{align}
\begin{split}\label{su2}
V_i^{\dagger}V_j &\coloneqq \begin{pmatrix}
\rho_{ij}e^{i\theta_{ij}} & \sqrt{1-\rho_{ij}^2}e^{i\varphi_{ij}} \\
-\sqrt{1-\rho_{ij}^2}e^{-i\varphi_{ij}}  & \rho_{ij}e^{-i\theta'_{ij}}
\end{pmatrix},\\ \rho_{14} &= \rho_{23} = 0.
\end{split}
\end{align}
Notice that, from Eqs.~\eqref{cond7}, $\rho_{ij}\neq 1$ and this implies:
\begin{equation}\label{sing_val2}
\theta_{ij}=\theta'_{ij},\ \alpha_{+i}\alpha_{+j} = \alpha_{-i}\alpha_{-j},\ \forall (h_i,h_j)\in S_+\times S_+.
\end{equation}
Using the equality $V_i^{\dagger}V_l = V_i^{\dagger}V_jV_j^{\dagger}V_l$, it is easy to show that
\begin{equation}\label{module}
\rho_{12} = \sqrt{1-\rho_{13}^2} = \sqrt{1-\rho_{24}^2} = \rho_{34},
\end{equation}
and, recalling Eqs.~\eqref{sing_val1},\eqref{sing_val2} and considering the determinants, one also realizes that
\begin{align*}
V_i^\dagger V_j \in \mathbb{SU}(2)\quad \forall(i,j)\in S_+\times S_+.
\end{align*}
In paricular, this holds $\forall \rho_{ij}\in [0,1]$ and one has $\theta_{ij} = \theta'_{ij}$ in Eq.~\eqref{su2}.

Accordingly, the only three possible cases are:
\begin{align}
\rho_{12} \neq 0,1;\label{i}\\
\rho_{12} = 1;\label{ii}\\
\rho_{12} = 0.\label{iii}
\end{align}
From paths of the form~\eqref{cond3} the following conditions follow:
\begin{multline}\label{first}
\alpha_{+1}\alpha_{-2}V_1 M V_2^{\dagger} + \alpha_{-1}\alpha_{+2}V_2 (1-M) V_1^{\dagger} + \\ + \alpha_{-3}\alpha_{+4}V_3 M V_4^{\dagger} + \alpha_{+3}\alpha_{-4}V_4 (1-M) V_3^{\dagger} = 0, \end{multline}
\begin{multline}\label{second}
\alpha_{+1}\alpha_{-2}M V_1^{\dagger} V_2 M + \alpha_{-1}\alpha_{+2}(1-M)V_2^{\dagger} V_1(1-M) + \\ + \alpha_{-3}\alpha_{+4}M V_3^{\dagger}V_4 M + \alpha_{+3}\alpha_{-4} (1-M) V_4^{\dagger}V_3 (1-M) = 0,
\end{multline}
\begin{multline}\label{third}
\alpha_{-1}\alpha_{+3}V_3 (1-M) V_1^{\dagger} + \alpha_{+1}\alpha_{-3}V_1 M V_3^{\dagger} + \\ + \alpha_{+2}\alpha_{-4}V_4 (1-M) V_2^{\dagger} + \alpha_{-2}\alpha_{+4}V_2 M V_4^{\dagger} = 0,
\end{multline}
\begin{multline}\label{fourth}
\alpha_{-1}\alpha_{+3}(1-M)V_3^\dagger V_1 (1-M) + \alpha_{+1}\alpha_{-3}MV_1^\dagger V_3M + \\ + \alpha_{+2}\alpha_{-4} (1-M) V_4^\dagger V_2 (1-M) + \alpha_{-2}\alpha_{+4} M V_2^\dagger V_4 M = 0,
\end{multline}
\begin{multline}\label{fifth}
\alpha_{-1}\alpha_{+4} MV_4^{\dagger}V_1 (1-M) + \alpha_{+1}\alpha_{-4} MV_1^\dagger V_4 (1-M) + \\ + \alpha_{+2}\alpha_{-3} MV_3^\dagger V_2 (1-M) + \alpha_{-2}\alpha_{+3} M V_2^{\dagger} V_3 (1-M) = 0.
\end{multline}
We shall use the previous equations in order to show that cases~\eqref{ii} and~\eqref{iii} lead to nontrivial solutions connected via a swap $2 \leftrightarrow 3$, while case~\eqref{i} does not satisfy unitarity.

\textbf{Case~\eqref{i}.} Recalling that in this case $\rho_{ij} \neq 1$ $\forall (i,j)$, then from Eqs.~\eqref{sing_val1},~\eqref{sing_val2} it is easy to derive that $\alpha_{+i} = \alpha_{-i} \eqqcolon \alpha_i\ \forall h_i\in S_+$. We can use this condition in Eq.~\eqref{second}, along with $\rho_{12} = \rho_{34}$, obtaining $\alpha_1\alpha_2 = \alpha_3\alpha_4$. Similarly, since $\rho_{13} = \rho_{24}$, from Eq.~\eqref{fourth} also $\alpha_1\alpha_3 = \alpha_2\alpha_4$ follows. One thus straightforwardly has: $\alpha_1 = \alpha_4$ and $\alpha_2 = \alpha_3$. Then again from Eqs.~\eqref{second} and~\eqref{fourth} it follows that $\ e^{i\theta_{12}} = -e^{i\theta_{34}},e^{i\theta_{13}} = -e^{i\theta_{24}}$. Finally, multiplying Eq.~\eqref{first} by $V_1^\dagger$ to the left and by $V_2$ to the right and using the two identities on the phases just found, the first matrix element reads:
\begin{align*}
&1-(1-\rho_{12}^2)-\rho_{13}^2 + e^{i(\theta_{14}-\theta_{23})} = 0\\ &\Longrightarrow\quad 2\rho_{12}^2 -1 = -e^{i(\theta_{14}-\theta_{23})},
\end{align*}
meaning that either $\rho_{12}=1$ or $\rho_{12}=0$, which is absurd.

\textbf{Case~\eqref{ii}.} Recalling Eq.~\eqref{module}, from Eq.~\eqref{second} one obtains $\alpha_{\pm 1}\alpha_{\mp 2} = \alpha_{\mp 3}\alpha_{\pm 4}$ and $V_1^\dagger V_2 = -V_3^\dagger V_4$, while from~\eqref{first} one gets
\begin{align}
\alpha_{\pm 1}\alpha_{\mp 2} &= \alpha_{\pm 3}\alpha_{\mp 4},\label{eq:alpha_id} \\
V_1^\dagger V_4 &= -V_2^\dagger V_3.\label{unitaries}
\end{align}
Plugging the latter equation into Eq.~\eqref{fifth}, one gets $\alpha_{+1}\alpha_{-4} + \alpha_{+2}\alpha_{-3} = \alpha_{-1}\alpha_{+4} + \alpha_{-2}\alpha_{+3}$. The latter equation found, using conditions~\eqref{eq:alpha_id}, reads
\begin{align*}
\frac{\alpha_{+1}\alpha_{-4}}{\alpha_{+1}\alpha_{-2}} + \frac{\alpha_{+2}\alpha_{-3}}{\alpha_{+4}\alpha_{-3}} = \frac{\alpha_{-1}\alpha_{+4}}{\alpha_{-1}\alpha_{+2}} + \frac{\alpha_{-2}\alpha_{+3}}{\alpha_{-4}\alpha_{+3}},
\end{align*}
which in turn implies:
\begin{align*}
\frac{\alpha_{+4}\alpha_{-4} + \alpha_{+2}\alpha_{-2}}{\alpha_{-2}\alpha_{+4}} = \frac{\alpha_{+4}\alpha_{-4} + \alpha_{+2}\alpha_{-2}}{\alpha_{+2}\alpha_{-4}},
\end{align*}
namely $\alpha_{-2}\alpha_{+4} = \alpha_{+2}\alpha_{-4}$. Combining this condition with those already found, we just end up with: $\alpha_{+i} = \alpha_{-i} \eqqcolon \alpha_i$ and $\alpha_1\alpha_2 = \alpha_3\alpha_4$.

Posing now $\theta \coloneqq \theta_{21}, \varphi \coloneqq \varphi_{31},V\coloneqq V_1$, we notice that $V^\dagger V_4 = -V_2^\dagger V_3 = -V_2^\dagger V V^\dagger V_3$ (where we used Eq.~\eqref{unitaries}) and that it can be set $e^{i\varphi}=1$ via a change of basis. We finally find the transition matrices:
\begin{equation}\label{transmatr2}
\begin{aligned}
A_{+1} &= \alpha_1 V \ket{0}\bra{0},\ & A_{-1} &= \alpha_1 V \ket{1}\bra{1}, \\ A_{+2} &= \alpha_2 e^{i\theta}V \ket{1}\bra{1},\ &A_{-2} &= \alpha_2e^{-i\theta} V \ket{0}\bra{0}, \\
A_{+3} &= -\alpha_3 V \ket{0}\bra{1},\ &A_{-3} &= \alpha_3 V \ket{1}\bra{0}, \\
A_{+4} &= -\alpha_4 e^{-i\theta}V \ket{1}\bra{0},\ &A_{-4} &= \alpha_4e^{i\theta} V \ket{0}\bra{1},
\end{aligned}
\end{equation}
along with the condition $\alpha_1\alpha_2 = \alpha_3\alpha_4$ and being $V$ an arbitrary unitary.

\textbf{Case~\eqref{iii}.} Similarly to the previous derivation, from Eq.~\eqref{fourth} we get $V_1^\dagger V_3 = -V_2^\dagger V_4$, while from Eq.~\eqref{first} we have $V_1^\dagger V_4 = V_2^\dagger V_3$. Then, from the form of $V_1^\dagger V_3$ and $V_1^\dagger V_4$, the identity $V_1^\dagger V_2 = V_1^\dagger V_3V_3^\dagger V_2 = V_1^\dagger V_3V_4^\dagger V_1 = V_4^\dagger V_1V_3^\dagger V_1$ holds. Moreover, using Eqs.~\eqref{third},~\eqref{fourth} and~\eqref{fifth}, one likewise can derive the conditions $\alpha_{+i} = \alpha_{-i} \eqqcolon \alpha_i$ and $\alpha_1\alpha_3 = \alpha_2\alpha_4$. At this point one easily realizes that this solution is connected to the previous one via a swap $2 \leftrightarrow 3$.

We can now check that the unitarity conditions
\begin{align*}
\begin{split}
A_{h}A_e^\dagger + A_eA_{-h}^\dagger = 0, \\
A_e^\dagger A_{h} + A_{-h}^\dagger A_e = 0
\end{split}
\end{align*}
cannot be satisfied for all $h\in S_+$.

Finally, we can impose condition~\eqref{eq:sum_unit}, finding
\begin{align}\label{eq:iso_unitary}
V = \begin{pmatrix}
\alpha_1+\alpha_2e^{i\theta} & \alpha_3-\alpha_4e^{i\theta} \\
-\alpha_3+\alpha_4e^{-i\theta} & \alpha_1+\alpha_2e^{-i\theta}
\end{pmatrix} \eqqcolon
\begin{pmatrix}
\beta & \gamma \\
-\gamma^* & \beta^*
\end{pmatrix}.
\end{align}
Plugging this into the transition matrices~\eqref{transmatr2}, one obtains:
\begin{equation}\label{transmatr'}
\begin{aligned}
A_{+1} &= \alpha_1\begin{pmatrix}
\beta & 0 \\ -\gamma^* & 0 
\end{pmatrix},\ 
A_{-1} = \alpha_1\begin{pmatrix}
0 & \gamma \\ 0 & \beta^*\end{pmatrix},\\
A_{+2} &= \alpha_2e^{i\theta}\begin{pmatrix}
0 & \gamma \\ 0 & \beta^*\end{pmatrix},\ 
A_{-2} = \alpha_2e^{-i\theta}\begin{pmatrix}
\beta & 0 \\ -\gamma^* & 0 
\end{pmatrix},\\
A_{+3} &= -\alpha_3\begin{pmatrix}
0 & \beta \\ 0 & -\gamma^*
\end{pmatrix},\ 
A_{-3} = \alpha_3\begin{pmatrix}
\gamma & 0 \\ \beta^* & 0\end{pmatrix},\\
A_{+4} &= -\alpha_4e^{-i\theta}\begin{pmatrix}
\gamma & 0 \\ \beta^* & 0\end{pmatrix},\ 
A_{-4} = \alpha_4e^{i\theta}\begin{pmatrix}
0 & \beta \\ 0 & -\gamma^*
\end{pmatrix},
\end{aligned}
\end{equation}
with $\alpha_i>0$, $\alpha_1\alpha_2 = \alpha_3\alpha_4$, $\sum_{i=1}^4\alpha^2_i=1$, $\beta$ and $\gamma$ defined as in Eq.~\eqref{eq:iso_unitary}, and up to a left-multiplication by an arbitrary unitary commuting with the symmetry group giving the invariance condition~\eqref{eq:isotropic_invariance}.

\bibliography{scalar}

\end{document}